\documentclass[11pt, twoside, letterpaper]{amsart}
\usepackage{amsmath, hyperref, color}
\usepackage{graphicx}
\usepackage{subfig}
\allowdisplaybreaks
\usepackage[normalem]{ulem}

\newtheorem{thm}{Theorem}[section]

\newtheorem{lem}[thm]{Lemma}
\newtheorem{prop}[thm]{Proposition}
\theoremstyle{definition}
\newtheorem{defn}[thm]{Definition}
\newtheorem{ass}[thm]{Assumption}
\theoremstyle{remark}
\newtheorem{rem}[thm]{Remark}

\numberwithin{equation}{section}
\newcommand{\R}{\mathbb{R}}
\newcommand{\C}{\mathbb{C}}
\newcommand{\N}{\mathbb{N}}

\newcommand{\ind}{\mathbf{1}}
\newcommand{\cA}{\mathcal{A}}	
\newcommand{\cB}{\mathcal{B}}	
	
\newcommand{\cF}{\mathcal{F}}
\newcommand{\cY}{\mathcal{Y}}
\newcommand{\cX}{\mathcal{X}}
\newcommand{\cD}{\mathcal{D}}
\newcommand{\FF}{\mathbb{F}}
\newcommand{\EE}{\mathbb{E}}

\newcommand{\QQ}{\mathbb{Q}}
\newcommand{\ud}{\mathrm d}
\newcommand{\tildeM}{\widetilde{M}}

\newcommand{\lois}{L^{{\rm OIS}}}
\newcommand{\rois}{r}
\newcommand{\im}{\ensuremath{\mathsf{i}}}
\renewcommand{\Re}{\mathrm{Re}}
\renewcommand{\Im}{\mathrm{Im}}
\newcommand{\dbra}[1]{[\kern-0.15em[ #1 ]\kern-0.15em]}
\newcommand{\dbraco}[1]{[\kern-0.15em[ #1 [\kern-0.15em[}
\newcommand{\dbraoc}[1]{]\kern-0.15em] #1 ]\kern-0.15em]}
\newcommand{\dbraoo}[1]{]\kern-0.15em] #1 [\kern-0.15em[}
\newcommand{\be}{\begin{equation}}
\newcommand{\ee}{\end{equation}}
\newcommand{\ba}{\begin{aligned}}
\newcommand{\ea}{\end{aligned}}

\begin{document}

\title{Multiple yield curve modelling with CBI processes}

\author[C. Fontana]{Claudio Fontana}
\address{Claudio Fontana, Department of Mathematics ``Tullio Levi Civita'', University of Padova (Italy)}
\email{fontana@math.unipd.it}

\author[A. Gnoatto]{Alessandro Gnoatto}
\address{Alessandro Gnoatto, Department of Economics, University of Verona (Italy)}
\email{alessandro.gnoatto@univr.it}

\author[G. Szulda]{Guillaume Szulda}
\address{Guillaume Szulda, Laboratoire de Probabilit\'es, Statistique et Mod\'elisation (LPSM), Universit\'e de Paris (France)}
\email{guillaume.szulda@upmc.fr}

\thanks{}
\thanks{We are thankful to two anonymous Reviewers for valuable remarks and suggestions that helped to improve the paper.
G. Szulda acknowledges hospitality and financial support from the University of Verona, where part of this work has been conducted.
Financial support from the University of Padova (research programme BIRD190200/19) and the Europlace Institute of Finance is gratefully acknowledged.}

\date{\today}


\maketitle

\begin{abstract}
We develop a modelling framework for multiple yield curves driven by continuous-state branching processes with immigration (CBI processes). Exploiting the self-exciting behavior of CBI jump processes, this approach can reproduce the relevant empirical features of spreads between different interbank rates. 
In particular, we introduce multi-curve models driven by a flow of tempered alpha-stable CBI processes. Such models are especially parsimonious and tractable, and can generate contagion effects among different spreads.
We provide a complete analytical framework, including a detailed study of discounted exponential moments of CBI processes. The proposed approach allows for explicit valuation formulae for all linear interest rate derivatives and semi-closed formulae for non-linear derivatives via Fourier techniques and quantization. 
We show that a simple specification of the model can be successfully calibrated to market data.
\end{abstract}

\section{Introduction}	\label{sec:intro}

The emergence of multiple yield curves can be rightfully regarded as the most relevant feature of interest rate markets over the last decade, starting from the 2007-2009 financial crisis. While pre-crisis interest rate markets were adequately described by a single yield curve and interbank rates (to which we generically refer as Ibor rates\footnote{The most relevant Ibor rates are represented by the Libor rates in the London interbank market and the Euribor rates in the Eurozone.}) associated to different tenors were determined by simple no-arbitrage relations, this proved to be no longer valid in the post-crisis scenario, where yield curves associated to interbank rates of different tenors exhibit a distinct behavior. 
This is reflected by the presence of tenor-dependent spreads between different yield curves. In the midst of the financial crisis, such spreads reached their peak beyond 200 basis points and since then, and still nowadays, they continue to remain at non-negligible levels (see Figure \ref{fig:spreads_analysis}). 
The credit, liquidity and funding risks existing in the interbank market, which were deemed negligible before the crisis, are at the origin of this phenomenon (see \cite{CD13,fitr12} in this regard).

\begin{figure}[ht]
\centering
\includegraphics[scale=0.45]{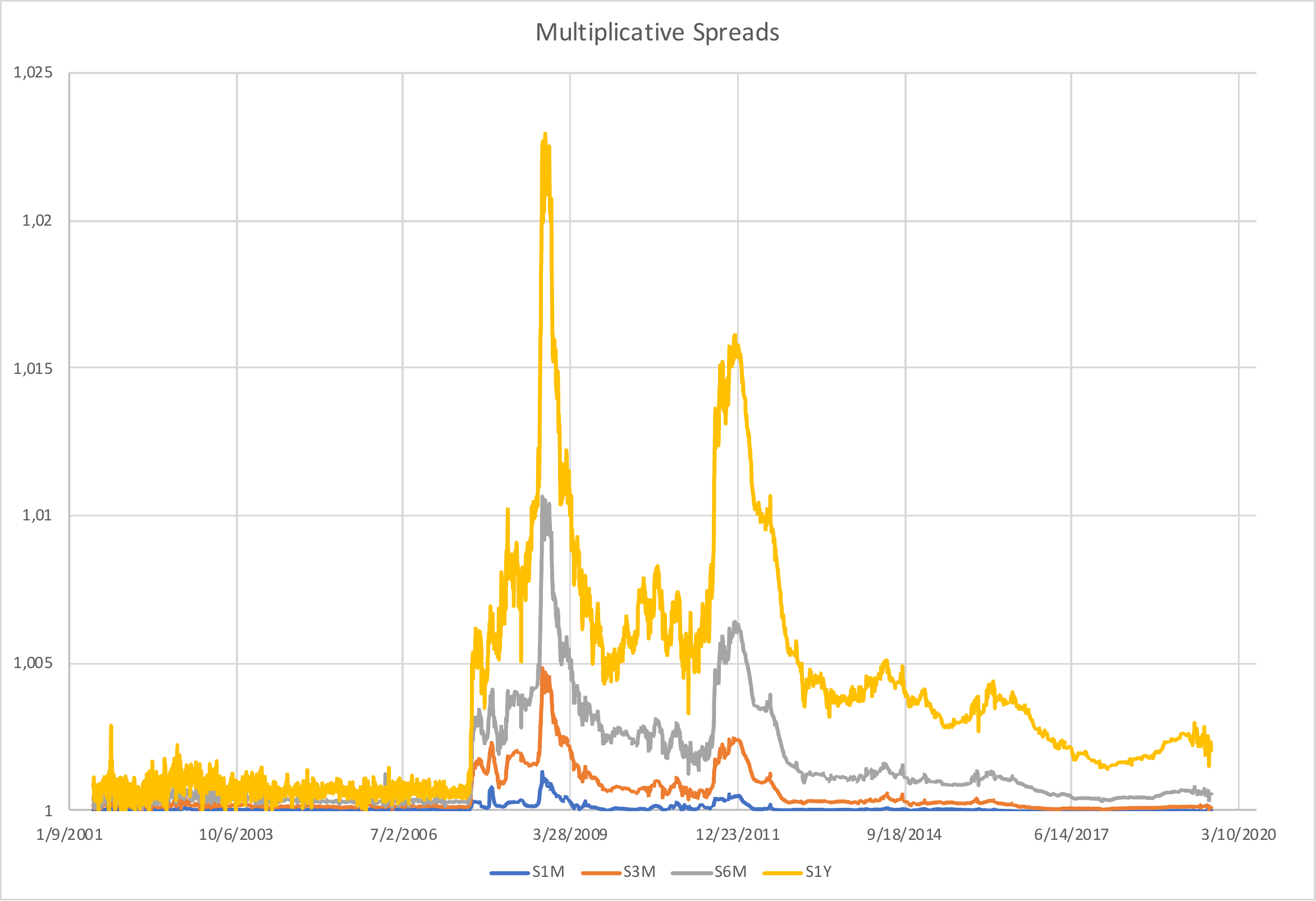}
\caption{Euribor-OIS spreads from 06/2001 to 09/2019.
Source: Bloomberg. \label{fig:spreads_analysis}}
\end{figure}

In all major economies, transaction-based backward-looking rates are currently being introduced as a replacement for Ibor rates (e.g., SOFR in the US market, ESTER in the Eurozone, SONIA in the UK market), also as a response to the 2012 Libor manipulation scandal.
At the time of writing, definitive conclusions on the evolution of Ibor rates cannot be drawn. However, there seems to be a consensus on the fact that the multiple curve framework will remain relevant (and, possibly, even more relevant) in the future. Indeed, in line with \cite{LM:19}, a complete disappearence of Ibor rates, reflecting the unsecured funding costs of banks, does not seem a realistic scenario. For instance, in the Eurozone the Euribor rate will not be abandoned, but only replaced by a reformed version in 2022. 
Moreover, Ibor proxies may arise to address the need for term rates  containing systemic credit or liquidity risk premia (see again \cite{LM:19}).

In this paper, we propose a novel modelling approach to multiple yield curves, ensuring analytical tractability as well as consistency with the most relevant empirical features. An inspection of Figure \ref{fig:spreads_analysis} reveals several important properties of spreads: first, spreads are typically greater than one and increasing with respect to the tenor; second, there are strong comovements (in particular, common upward jumps) among spreads associated to different tenors; third, relatively large values of the spreads are associated to high volatility, showing volatility clustering zones during crisis periods; fourth, low values of some spreads can persist for prolonged periods of time. To the best of our knowledge, a model that can adequately reproduce all these features does not yet exist.

By relying on the theory of {\em continuous-state branching processes with immigration} (CBI processes), we develop a modelling framework that can capture all the empirical properties mentioned above and, at the same time, allows for an efficient valuation of interest rate derivatives written on Ibor rates. 
Exploiting the affine property of CBI processes, we design our modelling framework in the context of the affine multi-curve models recently studied in \cite{CFGaffine}, taking multiplicative spreads and the OIS short rate as fundamental modelling objects. 
By construction, the model achieves a perfect fit to the observed term structures and can generate spreads greater than one and increasing with respect to the tenor.
The construction of the model requires a detailed study of the finiteness of  exponential moments of a CBI process. To this effect, we prove a general explicit characterization of the time of explosion of (discounted) exponential moments of a CBI process (see Section \ref{sec:prelim_CBI}), specializing to our context some techniques introduced in \cite{kr11}. This result can be considered of independent interest in the theory of CBI processes.

In the context of our general modelling framework, we introduce a tractable specification driven by a {\em flow of tempered alpha-stable CBI processes} (see Section \ref{sec:flow}). 
The adoption of a flow of CBI processes (see \cite{DL12}) enables us to capture strong comovements among spreads, including common upwards jumps and jump clustering effects. The characteristic self-exciting behavior of CBI processes proves to be a key ingredient to reproduce these features.
The choice of a tempered alpha-stable jump measure presents a good balance between flexibility and analytical tractability and allows for an explicit characterization of several important properties of the model. All linear interest rate derivatives admit closed pricing formulae, futures convexity adjustments can be explicitly computed and, by relying on Fourier techniques, we derive a semi-closed pricing formula for caplets. In addition, we develop a pricing method based on quantization, which is here applied for the first time to an interest rate setting. 
A specification of this model with two tenors is then calibrated to market data, showing an excellent fit to the data (see Section \ref{sec:calibration}).
We believe that the introduction of models driven by a flow of CBI processes can lead to further applications in other contexts where different term structures coexist.

We close this introduction by briefly discussing the related literature. We restrict our attention to the contributions that are specifically related to our work and do not attempt a general overview of multiple curve modelling, referring instead to the volumes \cite{BianchettiMorini13,Henr14,GR15} for detailed accounts on the topic. 
Our modelling approach adopts a short rate formulation. Short rate multi-curve models have been proposed in \cite{fitr12,GM:14,GMR:15,ken10,kitawo09,MR14} (see \cite[Chapter 2]{GR15} for a unifying treatment of these models) and are based on a representation of Ibor rates in terms of fictitious bond prices, which are computed by analogy to short rate models in the classical single-curve setting.
This results in tractable pricing models, but typically necessitates the modelling of quantities that are not observable on the market, with a consequent difficulty in capturing the stylized facts reported above.
An alternative short rate approach has been recently developed in \cite{CFGaffine}, without postulating the existence of fictitious bonds and modelling instead multiplicative spreads that can be directly inferred from market quotes (see Section \ref{sec:rates_spreads} for more details). 
In this work, we adopt the \cite{CFGaffine} approach. 
While \cite{CFGaffine} focused on the general theoretical properties of the modelling framework, we contribute by introducing a new class of tractable and flexible models that are specifically motivated by the empirical features discussed above.
The present paper is also related to the literature on CBI processes.
After their original application to population dynamics (see \cite{Pardoux} for an overview), CBI processes have been  adopted with success in finance, mainly due to their characteristic self-exciting behavior. Starting with the seminal work \cite{Fil01}, CBI processes have found a natural application in the context of interest rate modelling. In particular, in a single-curve interest rate model, \cite{JMS17} have shown that an alpha-stable CBI process allows to reproduce short rates with persistently low values. The same stochastic process has been used in \cite{JMSZ18} for stochastic volatility modelling, extending the classical model by Heston. CBI processes have been also applied to the modelling of forward prices in energy markets, where jump clustering phenomena are often observed, see  \cite{CMS19,JMSS18}. We also mention that, in a multiple curve setting, self-exciting features have been recently studied by \cite{HNL19} in a reduced-form model of interbank credit risk.

The paper is structured as follows. Section \ref{sec:general} presents some general results on CBI processes and the general modelling approach. This approach is then specialized in Section \ref{sec:flow} to a class of models driven by a flow of tempered alpha-stable CBI processes. Section \ref{sec:calibration} contains some numerical results, including calibration to market data, while Section \ref{sec:conclusions} concludes the paper. For the sake of readability, we postpone to Appendix \ref{sec:proofs} the proofs of the technical results stated in Sections \ref{sec:prelim_CBI} and \ref{sec:alpha_stable}. Appendix \ref{app:simulation} describes a simulation method for tempered alpha-stable CBI processes.

\section{General Modelling of Multiple Curves via CBI Processes}	\label{sec:general}

In this section, we develop a general modelling framework based on CBI processes for financial markets with multiple curves. To this effect, we adapt the affine short rate multi-curve approach of \cite{CFGaffine}, to which we refer for additional details on the general features of the post-crisis interest rate market. In this section, we focus on the construction and properties of the framework. The detailed analysis of an explicit specification is presented in Section \ref{sec:flow}.

\subsection{General properties of CBI processes}	\label{sec:prelim_CBI}

We start by providing some theoretical results which are relevant for the construction of multi-curve models driven by CBI processes. 
For ease of exposition, all proofs are postponed to Appendix \ref{sec:proofs}.
We refer the reader to \cite{Li}, \cite{Li19} and \cite[Chapter 10]{K06} for comprehensive accounts on CBI processes.
We start by recalling the general definition of a (conservative, stochastically continuous) CBI process, which has been first introduced in \cite{KW71}. 
We define the functions $\phi:\R_+\rightarrow\R$ and $\psi:\R_+\rightarrow\R_+$ by
\begin{align}
\phi(z) &:= bz + \frac{\sigma^2}{2}z^2 + \int_0^{+\infty}(e^{-zu}-1+zu)\pi(\ud u),	
\label{eq:branching}\\
\psi(z) &:= \beta z + \int_0^{+\infty}(1-e^{-zu})\nu(\ud u),
\label{eq:immigration}
\end{align}
for all $z\geq0$, where $(b,\sigma)\in\R^2$, $\beta\geq0$ and $\pi$ and $\nu$ are two sigma-finite measures on $(0,+\infty)$ such that $\int_0^{+\infty}(u\wedge u^2)\pi(\ud u)<+\infty$ and $\int_0^{+\infty}(1\wedge u)\nu(\ud u)<+\infty$, respectively.
For $p\geq0$, we also define the function $v(\cdot,p,0):\R_+\rightarrow\R_+$ as the unique non-negative solution to the ODE
\[
\frac{\partial}{\partial t}v(t,p,0) = -\phi\bigl(v(t,p,0)\bigr),
\qquad v(0,p,0) = p.
\]

\begin{defn}	\label{def:CBI}
A Markov process $X=(X_t)_{t\geq0}$ with initial value $X_0=x$ and state space $[0,+\infty)$ is a {\em continuous-state branching process with immigration} (CBI process) with {\em branching mechanism} $\phi$ and {\em immigration rate} $\psi$, denoted as CBI($\phi,\psi$), if its transition semigroup $(P_t)_{t\geq0}$ on $[0,+\infty)$ is defined by
\[
\int_{[0,+\infty)}e^{-p y}P_t(x,\ud y)
= \exp\left(-xv(t,p,0)-\int_0^t\psi\bigl(v(s,p,0)\bigr)\ud s\right),
\qquad\text{ for all }t\geq0.
\]
\end{defn}

A CBI process admits a representation as the solution to a certain  stochastic integral equation, which is especially useful for  numerical simulation purposes (see Appendix \ref{app:simulation}).
To this effect, let $(\Omega,\cF,\QQ)$ be a probability space endowed with a right-continuous filtration $\FF=(\cF_t)_{t\geq0}$, with respect to which all processes are assumed to be adapted. 
Let $W(\ud s, \ud u)$ be a white noise on $(0,+\infty)^2$ with intensity $\ud s\,\ud u$ and $M(\ud s,\ud z,\ud u)$ a Poisson time-space random measure on $(0,+\infty)^3$ with intensity $\ud s\,\pi(\ud z)\,\ud u$ (see \cite{Li,Li19}). The associated compensated random measure is denoted by $\tildeM(\ud s,\ud z,\ud u):=M(\ud s,\ud z,\ud u)-\ud s\,\pi(\ud z)\,\ud u$.
Let also $L=(L_t)_{t\geq0}$ be an increasing L\'evy process (subordinator) with $L_0=0$ and Laplace exponent $\psi$ as given in \eqref{eq:immigration}. By the L\'evy-It\^o decomposition, there exists a Poisson random measure $N(\ud s,\ud z)$ on $(0,+\infty)^2$ with intensity $\ud s\,\nu(\ud z)$ such that
$
L_t = \beta t + \int_0^t\int_0^{+\infty}zN(\ud s,\ud z)
$,
for all $t\geq0$.
We assume that $W$, $M$ and $N$ are independent.

For $x\geq0$, let us consider the following stochastic integral equation, referring to \cite[Section 7.3]{Li} for a detailed account of time-space random measures and the corresponding stochastic integrals:
\be	\label{eq:SDE_CBI}	\ba
X_t = x &+\int_0^t(\beta-b X_{s-})\ud s 
+ \sigma\int_0^t\int_0^{X_{s-}}W(\ud s,\ud u)\\ 
& + \int_0^t\int_0^{+\infty}\!\int_0^{X_{s-}}z\tildeM(\ud s,\ud z,\ud u) + \int_0^t\int_0^{+\infty}zN(\ud s,\ud z),
\qquad\text{ for all }t\geq0.
\ea	
\ee
The following result, which follows directly from \cite[Theorems 8.3 and 8.5]{Li19} and \cite[Theorem 3.1]{DL12}, provides the connection between CBI processes and the stochastic integral equation \eqref{eq:SDE_CBI}.

\begin{prop}	\label{prop:CBI}
A non-negative c\`adl\`ag process $X=(X_t)_{t\geq0}$ with $X_0=x$ is a CBI($\phi,\psi$) process if and only if it is a weak solution to \eqref{eq:SDE_CBI}. Moreover, for every $x\geq0$, equation \eqref{eq:SDE_CBI} admits a unique strong solution $X=(X_t)_{t\geq0}$ on $(\Omega,\cF,\FF,\QQ)$ with $X_0=x$ taking values in $[0,+\infty)$. 
\end{prop}

\begin{rem}	\label{rem:BM}
By L\'evy's characterization theorem, the process $B=(B_t)_{t\geq0}$ defined as
\[
B_t := \int_0^t\int_0^{X_{s-}}X_{s-}^{-1/2}\,\ind_{\{X_{s-}>0\}}W(\ud s,\ud u)
+ \int_0^t\int_0^1\ind_{\{X_{s-}=0\}}W(\ud s,\ud u),
\qquad\text{ for all }t\geq0,
\]
is a Brownian motion on $(\Omega,\mathcal{F},\mathbb{F},\mathbb{Q})$. The stochastic integral equation \eqref{eq:SDE_CBI} can then be equivalently rewritten replacing the term $\int_0^t\int_0^{X_{s-}}W(\ud s,\ud u)$ with the usual stochastic integral $\int_0^t\sqrt{X_{s}}\ud B_s$. This shows that CBI processes can be viewed as discontinuous generalizations of the classical square-root process, widely adopted for interest rate modelling.
The general representation \eqref{eq:SDE_CBI} will turn out to be necessary when considering a flow of CBI processes, as in Section \ref{sec:model_spec}.
\end{rem}

\begin{rem}	\label{rem:jump_cluster}
The stochastic integral equation \eqref{eq:SDE_CBI} makes evident the self-exciting behavior of a CBI process. 
Indeed, the two martingale components (i.e., the stochastic integrals with respect to $W$ and $\tildeM$) depend on the current value of the process itself and, therefore, large values of the process are associated to a relatively high volatility.
In particular, the jump intensity increases whenever a jump occurs, thereby generating jump clustering effects. 
These properties have a particularly relevant role for reproducing  the empirical features of spreads reported in Section \ref{sec:intro}.
\end{rem}

From the perspective of financial modelling, the analytical tractability of CBI processes is ensured by the fundamental and well-known link with affine processes (see \cite{dfs03,Fil01}). This is the content of the next result, which provides an analytical description of the joint Laplace transform of the process $X$ and its time integral $\int_0^{\cdot}X_s\,\ud s$. 
As a preliminary, let us define the convex set
\be	\label{eq:cY}
\cY := \biggl\{y\in\R : \int_{[1,+\infty)}e^{-yz}(\pi+\nu)(\ud z)<+\infty\biggr\} \supseteq \R_+.
\ee
In view of the standing assumptions on the measures $\pi$ and $\nu$, the set $\cY$ represents the set of values for which the functions $\phi$ and $\psi$ given in \eqref{eq:branching}-\eqref{eq:immigration} are finite-valued.
We define
\[
\ell := \inf\{y\in\R : \phi(y)<+\infty\}
\qquad\text{ and }\qquad
\kappa := \inf\{y\in\R : \psi(y)>-\infty\}.
\]
It can be easily verified that $\cY=[\ell\vee\kappa,+\infty)$ as long as  $\phi(\ell\vee\kappa)\vee(-\psi(\ell\vee\kappa))$ is finite (equivalently, $\int_{[1,+\infty)}e^{-(\ell\vee\kappa)z}(\pi+\nu)(\ud z)<+\infty$, provided that $\ell\vee\kappa>-\infty$), while $\cY=(\ell\vee\kappa,+\infty)$ otherwise.
For simplicity of presentation, we introduce the following mild technical assumption, which is assumed to be satisfied for the remaining part of Section \ref{sec:general}.

\begin{ass}	\label{ass:Lipschitz}
If $\ell\vee\kappa>-\infty$, then $\int_1^{+\infty}ze^{-(\ell\vee\kappa)z}\pi(\ud z)<+\infty$.
\end{ass}

It is well-known that the function $\phi$ is locally Lipschitz continuous on the interior $\cY^{\circ}$, but in general it may fail to be Lipschitz continuous at the boundary $\partial\cY$. 
Assumption \ref{ass:Lipschitz} corresponds to requiring that $\phi'(y)>-\infty$ for $y\in\partial\cY$ and thus ensures that $\phi$ is locally Lipschitz on the entire domain $\cY$.
In the proof of Theorem \ref{thm:affine} below, this assumption enables us to assert the existence of a unique solution to the ODE \eqref{eq:ODE_v}.
Note that Assumption \ref{ass:Lipschitz} is always satisfied by tempered $\alpha$-stable CBI processes with $\alpha\in(1,2)$, as considered in Section \ref{sec:alpha_stable}.

\begin{thm}	\label{thm:affine}
Let $X=(X_t)_{t\geq0}$ be a CBI($\phi,\psi$) process with $X_0=x$. Then $X$ is an affine process.
If Assumption \ref{ass:Lipschitz} holds, then, for every $(p,q)\in\cY\times\R_+$, the ODE
\be	\label{eq:ODE_v}
\frac{\partial}{\partial t}v(t,p,q) = q-\phi\bigl(v(t,p,q)\bigr),
\qquad v(0,p,q) = p,
\ee
admits a unique solution $v(\cdot,p,q):[0,T^{(p,q)})\rightarrow\cY$, where $T^{(p,q)}\in(0,+\infty]$, and it holds that
\be	\label{eq:aff_transform}
\EE\left[\exp\left(-pX_t - q\int_0^tX_s\,\ud s\right)\right]
= \exp\left(-xv(t,p,q) - \int_0^t\psi\bigl(v(s,p,q)\bigr)\ud s\right),
\ee
for all $t<T^{(p,q)}$, where $\phi$ and $\psi$ are defined as in \eqref{eq:branching}-\eqref{eq:immigration} on the extended domain $\cY$. 
\end{thm}

For $(p,q)\in\cY\times\R_+$, the time $T^{(p,q)}$ appearing in Theorem \ref{thm:affine} represents the maximum joint lifetime of $v(\cdot,p,q)$ and $\int_0^{\cdot}\psi(v(s,p,q))\ud s$. The lifetime $T^{(p,q)}$ characterizes the finiteness of (discounted) exponential moments, a crucial technical requirement for the modelling framework introduced in Section \ref{sec:framework}. By \cite[Proposition 3.3]{krm12} applied to the bi-dimensional affine process $(X,\int_0^{\cdot}X_s\,\ud s)$, it holds  that
\be	\label{eq:lifetime_moments}
T^{(p,q)} = \sup\left\{t\in\R_+ : \EE\bigl[e^{-pX_t-q\int_0^tX_s\,\ud s}\bigr]<+\infty\right\}.
\ee
In particular, we have that $\EE[\exp(-pX_t-q\int_0^tX_s\,\ud s)]<+\infty$, for all $t<T^{(p,q)}$. 
An explicit and general characterization of the lifetime $T^{(p,q)}$ is given in the next theorem. 

\begin{thm}	\label{thm:lifetime}
Suppose that Assumption \eqref{ass:Lipschitz} holds and let $p\in\cY$. For $q\in\R_+$, define the quantity $p_q:=\inf\{y\in\cY:q-\phi(y)\geq0\}$. If $p\geq p_q$, then it holds that $T^{(p,q)}=+\infty$. Otherwise, if $p< p_q$, then
\be	\label{eq:lifetime}
T^{(p,q)} = \int_{\ell\vee\kappa}^p\frac{\ud y}{\phi(y)-q}.
\ee
Suppose furthermore that $\psi(\ell\vee\kappa)>-\infty$. Then, $T^{(p,q)}=+\infty$ holds for all $(p,q)\in[\ell\vee\kappa,+\infty)\times\R_+$ if and only if $-\infty<\phi(\ell\vee\kappa)\leq0$.
\end{thm}

\begin{rem}	\label{rem:complex}
(1) A general representation of the lifetime of exponential moments of affine processes is given in \cite[Theorem 4.1]{kr11}. For the specific case of CBI processes, this general result is refined by our Theorem \ref{thm:lifetime}. Indeed, \cite[Theorem 4.1]{kr11} requires the validity of additional assumptions, which in particular only allow for CBI processes with a strictly subcritical branching mechanism.

(2) To price non-linear derivatives (see Section \ref{sec:caplet_pricing}), an extension of the affine transform formula \eqref{eq:aff_transform} to the complex domain is needed. To this effect, let $S(\cY^{\circ}):=\{p\in\C :\Re(p)\in\cY^{\circ}\}$, with $\Re(p)$ denoting the real part of $p$. For every $(p,q)\in S(\cY^{\circ})\times\R_+$, Theorem \ref{thm:affine} yields the existence of a unique solution $v(\cdot,\Re(p),q)$ to the ODE \eqref{eq:ODE_v} with initial value $v(0,\Re(p),q)=\Re(p)$ up to a lifetime $T^{(\Re(p),q)}$. By \cite[Theorem 2.26]{krm12}, if $T$ is such that $T<T^{(\Re(p),q)}$ and $v(t,\Re(p),q)\in\cY^{\circ}$ for all $t\in[0,T]$, then the affine transform formula \eqref{eq:aff_transform} holds for $p\in\C$ for all $t\in[0,T]$, replacing $\phi$ and $\psi$ by their analytic extensions to the complex domain $S(\cY^{\circ})$ (see \cite[Proposition 2.21]{krm12}).
In particular, as a consequence of Theorem \ref{thm:lifetime}, the affine transform formula \eqref{eq:aff_transform} is always valid for all $p\in\C$ such that $\Re(p)\in[p_q,+\infty)\cap\cY^{\circ}$.
\end{rem}

\subsection{OIS rates, Ibor rates and multiplicative spreads}	\label{sec:rates_spreads}

In this section, we introduce the fundamental quantities that will be modelled in Section \ref{sec:framework}.
In fixed income markets, the reference rates for overnight transactions are the EONIA (Euro overnight index average) rate in the Eurozone and the Federal Funds rate in the US market. The Eonia and the Federal Funds rates are determined on the basis of overnight transactions and are the underlying of overnight indexed swaps (OIS).
The term structure of OIS discount factors at time $t$ is represented by the map $T\mapsto B(t,T)$, where $B(t,T)$ denotes the price at  time $t$ of an OIS zero-coupon bond with maturity $T$, stripped from the market swap rates of OIS (see, e.g., \cite{GR15}).
We denote by $\rois_t$ the {\em OIS short rate}, defined as the short end of the term structure of instantaneous forward rates implied by OIS zero-coupon bond prices.
In market practice, the OIS short rate is typically approximated by the overnight rate associated to the shortest available tenor and is often adopted as a collateral rate.
The simply compounded OIS spot rate for the period $[t,t+\delta]$ is defined as
\be	\label{eq:lois}
\lois(t,t,\delta) := \frac{1}{\delta}\left(\frac{1}{B(t,t+\delta)}-1\right),
\qquad\text{ for }\delta\geq0\text{ and }t\geq0,
\ee
representing the swap rate of an OIS with a single cashflow at time $t+\delta$, evaluated at time $t$.
Note that the right-hand side of \eqref{eq:lois} corresponds to the pre-crisis textbook definition of Ibor rate.

Ibor rates are the underlying rates of fixed-income derivatives and are determined by a panel of primary financial institutions for unsecured lending.
We denote by $L(t,t,\delta)$ the {\em (spot) Ibor rate} for the time interval $[t,t+\delta]$ fixed at time $t$, where the tenor $\delta$ is typically one day (1D), one week (1W), or several months (1M, 2M, 3M, 6M, 12M). We consider Ibor rates for a generic set $\cD:=\{\delta_1,\ldots,\delta_m\}$ of tenors, with $0<\delta_1<\ldots<\delta_m$, for some $m\in\N$.
In the post-crisis environment Ibor rates associated to different tenors exhibit a distinct behavior and are no longer determined by simple no-arbitrage relations. As mentioned in the introduction, this leads to non-negligible basis spreads and to the emergence of multiple yield curves.

Our main modelling quantities are the {\em spot multiplicative spreads}
\be	\label{eq:mult_spread}
S^{\delta}(t,t) := \frac{1+\delta L(t,t,\delta)}{1+\delta \lois(t,t,\delta)},
\qquad\text{ for all }\delta\in\cD\text{ and }t\geq0,
\ee
together with the OIS short rate $\rois_t$.
In the post-crisis environment, multiplicative spreads are usually greater than one and increasing with respect to the tenor. Abstracting from liquidity and funding issues, this is due to the fact that Ibor rates embed the risk that the average credit quality of the bank panel deteriorates over the term of the loan, while OIS rates reflect the credit quality of a newly refreshed panel (see, e.g., \cite{CDS:01,fitr12}). 

As shown in \cite{CFG:16,CFGaffine}, modelling the OIS short rate $\rois_t$ together with the multiplicative spreads $\{S^{\delta}(t,t):\delta\in\cD\}$ suffices to provide a complete description of an interest rate market where the following two sets of assets are traded:
\begin{itemize}
\item OIS zero-coupon bonds, for all maturities $T>0$;
\item forward rate agreements (FRAs), for all maturities $T>0$ and for all tenors $\delta\in\cD$.
\end{itemize}
We recall that a FRA written on the Ibor rate $L(T,T,\delta)$ with rate $K$ is a contract which delivers the payoff $\delta(L(T,T,\delta)-K)$ at maturity $T+\delta$.
The {\em forward Ibor rate} $L(t,T,\delta)$ is defined for $t\leq T$ and $\delta\in\cD$ as the rate $K$ that makes equal to zero the value at time $t$ of a FRA written on $L(T,T,\delta)$ with rate $K$. 
Among all financial derivatives written on Ibor rates, FRAs can be  regarded as the basic building blocks, due to the fact that all linear interest rate products such as interest rate swaps and basis swaps can be represented as portfolios of FRAs (see \cite[Appendix A.1]{CFGaffine}).

As in \cite[Section 3.3]{CFGaffine}, we adopt a martingale approach and directly define the model on the filtered probability space $(\Omega,\cF,\FF,\QQ)$, where $\QQ$ is a probability measure such that all traded assets considered above are martingales under $\QQ$ when discounted by the OIS bank account $\exp(\int_0^{\cdot}\rois_s\ud s)$.
This implies that OIS zero-coupon bond prices can be expressed as
\be	\label{eq:gen_bond_price}
B(t,T) = \EE\Bigl[e^{-\int_t^T\rois_s\ud s}\Bigl|\cF_t\Bigr],
\qquad\text{ for all }0\leq t\leq T<+\infty,
\ee
and forward Ibor rates are given by
\be	\label{eq:gen_FRA_rate}
L(t,T,\delta) = \EE^{T+\delta}[L(T,T,\delta)|\cF_t],
\qquad\text{ for all }\delta\in\cD\text{ and }0\leq t\leq T<+\infty,
\ee
where $\EE^{T+\delta}$ denotes the expectation under the $(T+\delta)$-forward probability measure $\QQ^{T+\delta}$.

The idea of modelling multi-curve interest rate markets via multiplicative spreads is due to M. Henrard (see \cite{Henr14}) and has been recently pursued in \cite{NS15,CFG:16,CFGaffine,EGG18,FGGS20}.
Multiplicative spreads can be directly inferred from quoted Ibor and OIS rates and admit a natural economic interpretation. Indeed, $S^{\delta}(t,t)$ can be regarded as a market expectation (at date $t$) of the riskiness of the Ibor panel over the period $[t,t+\delta]$. As shown in \cite[Appendix B]{CFG:16}, this interpretation can be made precise via a foreign exchange analogy (see also \cite{MM18}).
Furthermore, in comparison to additive spreads (as considered for instance in \cite{mer10,merxie12}), multiplicative spreads represent a particularly tractable modelling quantity in relation with CBI processes.

\subsection{Modelling framework}	\label{sec:framework}

In this section, we present a general modelling framework for the OIS short rate $(\rois_t)_{t\geq0}$ and spot multiplicative spreads $\{(S^{\delta}(t,t))_{t\geq0};\delta\in\cD\}$ based on CBI processes. 

We assume that the filtered probability space $(\Omega,\cF,\FF,\QQ)$ supports a $d$-dimensional process $X=(X_t)_{t\geq0}$ such that each component $X^j$ is a CBI process with branching mechanism $\phi^j$ and immigration rate $\psi^j$, for $j=1,\ldots,d$. 
We assume that $X^1,\ldots,X^d$ are mutually independent.

Besides the driving process $X$, we introduce the following modelling ingredients:
\begin{enumerate}
\item[(i)] a function $\ell:\R_+\rightarrow\R$ such that $\int_0^T|\ell(u)|\ud u<+\infty$, for all $T>0$;
\item[(ii)] a vector $\lambda\in\R^d_+$;
\item[(iii)] a family of functions $\mathbf{c}=(c_1,\ldots,c_m)$, with $c_i:\R_+\rightarrow\R$ for all $i=1,\ldots,m$;
\item[(iv)] a family of vectors $\boldsymbol{\gamma}=(\gamma_1,\ldots,\gamma_m)$, with $\gamma_i\in\R^d$ for all $i=1,\ldots,m$.
\end{enumerate}

\begin{defn}	\label{def:general_model}
The tuple $(X,\ell,\lambda,\mathbf{c},\boldsymbol{\gamma})$ is said to generate a {\em CBI-driven multi-curve model} if
\begin{align}
\rois_t &= \ell(t) + \lambda^{\top}X_t,	
\label{eq:general_rate}\\
\log S^{\delta_i}(t,t) &= c_i(t) + \gamma_i^{\top}X_t,
\label{eq:general_spread}
\end{align}
for all $t\geq0$ and $i=1,\ldots,m$, and if the following conditions hold:
\be	\label{eq:general_model_life}
-\gamma_{i,j}\in\cY^j
\qquad\text{ and }\qquad
T^{(-\gamma_{i,j},\lambda_j)}=+\infty,
\qquad\text{ for all }i=1,\ldots,m\text{ and }j=1,\ldots,d,
\ee
where the set $\cY^j$ is defined as in \eqref{eq:cY} with respect to the CBI process $X^j$ and $T^{(-\gamma_{i,j},\lambda_j)}$ denotes the lifetime as in Theorem \ref{thm:affine} for the process $X^j$, with $p=-\gamma_{i,j}$ and $q=\lambda_j$.
\end{defn}

Condition \eqref{eq:general_model_life} serves to guarantee that $\EE[\exp(-\int_0^Tr_s\,\ud s)S^{\delta}(T,T)]<+\infty$, for all $T>0$ and $\delta\in\cD$, thus ensuring that the model can be applied to arbitrarily large maturities (i.e., the expected value in \eqref{eq:gen_FRA_rate} is always well-defined).
The role of the time-dependent functions $\ell$ and $\mathbf{c}$ consists in allowing the model to perfectly fit the observed term structures  (see \cite[Proposition 3.18]{CFGaffine} for a precise characterization of this property).
A multi-curve model constructed as in Definition \ref{def:general_model} inherits the properties of the CBI process $X$, in particular its jump clustering behavior (see Remark \ref{rem:jump_cluster}). 
Moreover, a CBI-driven multi-curve model can easily generate common upward jumps in different spreads. Indeed, in view of specification \eqref{eq:general_spread}, this can be  achieved by letting $\gamma^{\top}_i\gamma_j\neq0$, for $i,j=1,\ldots,m$ with $i\neq j$, meaning that the spreads associated to tenors $\delta_i$ and $\delta_j$ are affected by common risk factors. As mentioned in Section \ref{sec:intro}, common upward jumps represent a particularly important stylized fact.
We refer to Section \ref{sec:flow} for a more specific discussion of the adequacy of this approach in reproducing the empirical features of spreads mentioned in Section \ref{sec:intro}. 

\begin{rem}	\label{rem:number_factors}
In general, there are no constraints on the choice of the dimension $d$ of the driving process $X$. On the one hand, $d\geq m$ is needed to ensure non-trivial correlation structures among the $m$ spreads. On the other hand, the case $d<m$ is in line with market practice, which often assumes for simplicity the existence of linear (possibly time-varying) dependence among different spreads.
Let us also mention that models driven by a vector of independent CBI processes have been recently applied to spot and forward energy prices in \cite{CMS19} and \cite{JMSS18}.
\end{rem}

As shown in \cite{CFG:16,CFGaffine}, the basic building blocks for the valuation of interest rate derivatives in a multi-curve setting are represented by OIS zero-coupon bond prices and {\em forward multiplicative spreads} $S^{\delta}(t,T)$, defined as follows (compare with equation \eqref{eq:mult_spread}):
\be	\label{eq:fwd_spread}
S^{\delta}(t,T) := \frac{1+\delta L(t,T,\delta)}{1+\delta\lois(t,T,\delta)},
\qquad\text{  for $\delta\in\cD$ and $0\leq t\leq T<+\infty$},
\ee
where $L(t,T,\delta)$ is the forward Ibor rate and  $\lois(t,T,\delta)$ is the simply compounded OIS forward rate defined by $\lois(t,T,\delta):=(B(t,T)/B(t,T+\delta)-1)/\delta$.
It can be easily checked that \eqref{eq:gen_FRA_rate} implies that the forward multiplicative spread process $(S^{\delta}(t,T))_{t\in[0,T]}$ is a martingale under the $T$-forward measure $\QQ^T$, for every $\delta\in\cD$ and $T>0$ (compare also with \cite[Lemma 3.11]{CFG:16}).

The following proposition shows that in a CBI-driven multi-curve model OIS zero-coupon bond prices and forward multiplicative spreads can be computed in closed form. This represents a fundamental property in view of the practical applicability of our framework.
Due to condition \ref{eq:general_model_life} and the independence of the processes $(X^1,\ldots,X^d)$, the proposition follows by a direct application of Theorem \ref{thm:affine} together with \eqref{eq:gen_bond_price} and the martingale property of $(S^{\delta}(t,T))_{t\in[0,T]}$ under the $T$-forward  measure $\QQ^T$, for all $\delta\in\cD$ and $T>0$.
For $j=1,\ldots,d$ and $(p,q)\in\cY^j\times\R_+$, we denote by $v^j(t,p,q)$ the solution to the ODE \eqref{eq:ODE_v} with branching mechanism $\phi^j$.

\begin{prop}	\label{prop:bond_spreads}
Let $(X,\ell,\lambda,\mathbf{c},\boldsymbol{\gamma})$ generate a CBI-driven multi-curve model. Then:
\begin{enumerate}
\item[(i)] 
for all  $0\leq t\leq T<+\infty$, the OIS zero-coupon bond price $B(t,T)$ is given by
\be	\label{eq:general_bond}
B(t,T) = \exp\left(\cA_0(t,T) + \cB_0(T-t)^{\top}X_t\right),
\ee
where the functions $\cA_0(t,T)$ and $\cB_0(T-t)=(\cB_0^1(T-t),\ldots,\cB_0^d(T-t))^{\top}$ are given by
\begin{align*}
\cA_0(t,T) &:= -\int_0^{T-t}\Bigl(\ell(t+s)+\sum_{j=1}^d\psi^j\bigl(v^j(s,0,\lambda_j)\bigr)\Bigr)\ud s,	\\
\cB_0^j(T-t) &:= -v^j(T-t,0,\lambda_j),
\qquad\text{ for }j=1,\ldots,d;
\end{align*}
\item[(ii)]
for all  $0\leq t\leq T<+\infty$ and $i=1,\ldots,m$, the multiplicative spread $S^{\delta_i}(t,T)$ is given by
\be	\label{eq:general_fwdspread}
S^{\delta_i}(t,T) = \exp\left(\cA_i(t,T) + \cB_i(T-t)^{\top}X_t\right),
\ee
where the functions $\cA_i(t,T)$ and $\cB_i(T-t)=(\cB_i^1(T-t),\ldots,\cB_i^d(T-t))^{\top}$ are given by
\begin{align*}
\cA_i(t,T) &:= c_i(T) + \sum_{j=1}^d\int_0^{T-t}\Bigl(\psi^j\bigl(v^j(s,0,\lambda_j)\bigr)-\psi^j\bigl(v^j(s,-\gamma_{i,j},\lambda_j)\bigr)\Bigr)\ud s,\\
\cB_i^j(T-t) &:= v^j(T-t,0,\lambda_j)-v^j(T-t,-\gamma_{i,j},\lambda_j),
\qquad\text{ for }j=1,\ldots,d.
\end{align*}
\end{enumerate}
\end{prop}

Linear fixed-income products, such as forward rate agreements, interest rate swaps, basis swaps, can be priced in closed form by relying on the explicit expressions for OIS zero-coupon bond prices and forward multiplicative spreads given in Proposition \ref{prop:bond_spreads} together with the valuation formulae stated in \cite[Appendix A]{CFGaffine}.
Non-linear derivatives such as caps, floors and swaptions can be efficiently priced via Fourier techniques, as illustrated in Sections \ref{sec:caplet_pricing}-\ref{sec:caplet_pricing_quant} in the case of caplets. 

\begin{rem}[Futures convexity adjustments]	\label{rem:convexity}
A further advantage of CBI-driven multi-curve models consists in the possibility of computing in closed form {\em futures convexity adjustments}. We recall that the futures convexity adjustment $C(t,T,\delta)$ is defined as the difference at time $t$ between future and forward Ibor rates for the same reference period $[T,T+\delta]$ (see \cite{GM10,mer_futures}). 
More specifically,
\[
C(t,T,\delta_i)  
:= \EE[L(T,T,\delta_i)|\cF_t] -  \EE^{T+\delta_i}[L(T,T,\delta_i)|\cF_t]
= \EE[L(T,T,\delta_i)|\cF_t] - L(t,T,\delta_i),
\]
for $i=1,\ldots,m$ and $0\leq t\leq T<+\infty$, where the second equality follows from \eqref{eq:gen_FRA_rate}. 
The forward Ibor rate $L(t,T,\delta_i)$ can be directly obtained from \eqref{eq:lois} and \eqref{eq:fwd_spread} together with Proposition \ref{prop:bond_spreads}. By applying the affine transform formula \eqref{eq:aff_transform} again with Proposition \ref{prop:bond_spreads}, the future Ibor rate can be explicitly computed as
\[
\EE[L(T,T,\delta_i)|\cF_t]
= \frac{1}{\delta_i}\left(e^{c_i(T)-\cA_0(T,T+\delta_i)-\sum_{j=1}^d\int_0^{T-t}v(s,\cB_0^j(\delta_i)-\gamma_{i,j},0)\ud s - \sum_{j=1}^dv(T-t,\cB_0^j(\delta_i)-\gamma_{i,j},0)X^j_t}-1\right).
\]
\end{rem}

In typical market scenarios, multiplicative spreads are greater than one and increasing with respect to the tenor. As shown in the following proposition, these features can be easily reproduced.
While this result can be recovered as a special case of the general statement in \cite[Proposition 3.7]{CFGaffine}, we provide a short self-contained proof that relies on the specific properties of our setting.

\begin{prop}	\label{prop:general_properties}
Let $(X,\ell,\lambda,\mathbf{c},\boldsymbol{\gamma})$ generate a CBI-driven multi-curve model. Then, for every $i=1,\ldots,m$, the following hold: 
\begin{enumerate}
\item[(i)] 
if $\gamma_i\in\R^d_+$ and $c_i(t)\geq0$, for all $t\geq0$, then $S^{\delta_i}(t,T)\geq 1$ a.s. for all $0\leq t\leq T<+\infty$;
\item[(ii)] 
if $\gamma_{i+1}-\gamma_i\in\R^d_+$ and $c_i(t)\leq c_{i+1}(t)$, for all $t\geq0$, then $S^{\delta_i}(t,T)\leq S^{\delta_{i+1}}(t,T)$ a.s. for all $0\leq t\leq T<+\infty$.
\end{enumerate}
\end{prop}
\begin{proof}
Arguing similarly as in \cite[Proposition 3.1]{Li}, it can be shown that, for every $j=1,\ldots,d$, $q\in\R_+$ and $t\geq0$, the function $\cY^j\ni p\mapsto v^j(t,p,q)$ is strictly increasing.
Moreover, by \eqref{eq:immigration}, each immigration rate $\psi^j$ is an increasing function. 
By relying on these facts and since $X$ takes values in $\R^d_+$, the result follows as a direct consequence of part (ii) of Proposition \ref{prop:bond_spreads}.
\end{proof}

\begin{rem}[On the possibility of negative rates]	\label{rem:negativeOIS}
In recent years, negative short rates have been observed to coexist with spreads which are greater than one. 
Since the function $\ell$ in \eqref{eq:general_rate} is allowed to take negative values, our framework does not exclude this possibility. Moreover, a slight extension of Definition \ref{def:general_model} permits to generate OIS short rates which are not bounded from below by the deterministic function $\ell$. Indeed, it suffices to replace $X$ with a $(d+1)$-dimensional process $X'=(X,Y)$ such that $X'$ is an affine process and $\QQ(Y_t<0)>0$, for all $t\geq0$. Specification \eqref{eq:general_rate} can then be replaced by $\rois_t=\ell(t)+\lambda^{\top}X_t + Y_t$, while multiplicative spreads are given as in \eqref{eq:general_spread}. Note that $Y$ is not restricted to be independent of $X$.
A simple extension of this type has been tested in our calibration to market data (see Section \ref{sec:calibration_results} below).
\end{rem}

\section{A Tractable Specification via Tempered Alpha-Stable CBI Processes}	\label{sec:flow}

In this section, we introduce a multi-curve model driven by a {\em flow of tempered $\alpha$-stable CBI processes}. The proposed specification is parsimonious, captures the most relevant features of post-crisis interest rate markets and, at the same time, allows for efficient pricing of caplets via Fourier and quantization techniques, as shown in Sections \ref{sec:caplet_pricing}-\ref{sec:caplet_pricing_quant}.
The empirical performance of this specification will be studied in Section \ref{sec:calibration} by calibration to market data.
We start by introducing the probabilistic properties of the class of tempered $\alpha$-stable CBI processes. For better readability, the proofs of the results stated in Section \ref{sec:alpha_stable} are deferred to Appendix \ref{sec:proofs}.

\subsection{Tempered alpha-stable CBI processes}	\label{sec:alpha_stable}

The specific features of a CBI process are determined by its branching mechanism $\phi$ and immigration rate $\psi$, notably by the measures $\pi$ and $\nu$ appearing in \eqref{eq:branching}-\eqref{eq:immigration}. In the following definition, we introduce a tractable and flexible specification, which is particularly well-suited to the modelling of multiple yield curves.

\begin{defn}	\label{def:alpha_stable}
A CBI($\phi,\psi$) process $X=(X_t)_{t\geq0}$ with $X_0=x$ is called a {\em tempered $\alpha$-stable} CBI process if $\nu(\ud u)=0$ and
\be	\label{eq:alpha_stable}
\pi(\ud u) = C\frac{e^{- \theta u}}{u^{1+\alpha}}\ind_{\{u>0\}}\ud u,
\ee
where $\theta>0$,  $\alpha<2$ and $C$ is a suitable normalizing constant.
\end{defn}

For a tempered $\alpha$-stable CBI process, it holds that $\cY=[-\theta,+\infty)$. Indeed, since $\nu(\ud u)=0$, we have that $\int_1^{+\infty}e^{\theta u}\pi(\ud u)=C/\alpha$, while $\int_1^{+\infty}e^{(\theta+\epsilon)u}\pi(\ud u)=+\infty$ for every $\epsilon>0$.
The measure $\pi$ given in \eqref{eq:alpha_stable} corresponds to the L\'evy measure of a spectrally positive (generalized) tempered $\alpha$-stable compensated L\'evy process $Z=(Z_t)_{t\geq0}$, whose characteristic function is given by
\[
\EE[e^{\im uZ_t}]
= \exp\left(C\,\Gamma(-\alpha)\,\theta^{\alpha}\left(\left(1-\frac{\im u}{\theta}\right)^{\alpha}-1+\alpha\frac{\im u}{\theta}\right)t\right),
\]
for all $u\in\R$ and $t\geq0$, as long as $\alpha\notin\{0,1\}$, see \cite[Proposition 4.2]{ContTankov}, where $\Gamma$ denotes the Gamma function extended to $\R\setminus\mathbb{Z}_-$ (see \cite[Section 1.1]{Lebedev}).
The process $Z$ has different path properties depending on the choice of the parameter $\alpha$:
\begin{itemize}
\item if $\alpha<0$, then $Z$ is a compensated compound Poisson process, since $\pi(\R_+)<+\infty$;
\item if $\alpha\in[0,1)$, then $Z$ has infinite activity and finite variation, since $\int_0^1u\,\pi(\ud u)<+\infty$;
\item if $\alpha\in[1,2)$, then $Z$ has infinite activity and infinite variation.
\end{itemize}
Choosing $\alpha=0$ yields the L\'evy measure of a Gamma subordinator, while $\alpha=1/2$ corresponds to an inverse Gaussian subordinator.
Small positive values of $\alpha$ are likely to generate jump clustering effects. Indeed, in view of \eqref{eq:alpha_stable}, small positive values of $\alpha$ imply that the random measure $M$ is more likely to generate large jumps. Since the jump intensity depends on the current value of the process $X$ (self-exciting behavior), large jumps increase the likelihood of further jumps, thus generating jump clustering phenomena (compare with Remark \ref{rem:jump_cluster}).
The tempering parameter $\theta$ determines the tail behavior of the jump measure $\pi$ and will be indispensable to ensure the finiteness of exponential moments (compare with Proposition \ref{prop:properties} and Remark \ref{rem:non_tempered} below).

In the following section, we shall focus on the case $\alpha\in(1,2)$. In this case, it is convenient to choose the normalizing constant $C$  as follows, for some $\eta>0$:
\be	\label{eq:C}
C(\alpha,\eta) := -\frac{\eta^{\alpha}}{\Gamma(-\alpha)\cos(\alpha\pi/2)}.
\ee
As will become clear in Section \ref{sec:model_spec}, the constant $\eta$ will play the role of a volatility parameter determining the impact of jumps in the dynamics of the model.

Similarly as in \cite{JMS17}, small values of $\alpha$ in the range $(1,2)$ increase the likelihood that the process $X$ spends prolonged periods of time at relatively low levels. Indeed, a smaller value of $\alpha$ implies a stronger compensation effect in $\tildeM$, corresponding to a stronger negative drift after a large jump of $X$. 
In turn, since the jump intensity is proportional to the current level of $X$, this leads to a sharp reduction of the jump intensity, thus increasing the likelihood of a persistence of low values for the process $X$.
This effect is especially important in view of our modelling purposes, since spreads and rates tend to exhibit prolonged periods of low values, as mentioned in the introduction.

Since $\nu(\ud u)=0$ (see Definition \ref{def:alpha_stable}), the immigration rate $\psi$ of a tempered $\alpha$-stable CBI process reduces to $\psi(z)=\beta z$. The branching mechanism $\phi$ is explicitly described in the following lemma.

\begin{lem}	\label{lem:branch_alpha}
For $\theta\geq0$, $\alpha\in(1,2)$ and $C=C(\alpha,\eta)$, the branching mechanism $\phi$ of a tempered $\alpha$-stable CBI process is a convex function on $[-\theta,+\infty)$ and is explicitly given by
\be	\label{eq:branch_alpha}
\phi(z) = bz + \frac{\sigma^2}{2}z^2 
+ \eta^{\alpha}\frac{\theta^{\alpha}+z\alpha\theta^{\alpha-1}-(z+\theta)^{\alpha}}{\cos(\alpha\pi/2)},
\qquad\text{ for all }z\geq-\theta.
\ee
The branching mechanism $\phi$ is decreasing with respect to the tempering parameter $\theta$.
Moreover, Assumption \ref{ass:Lipschitz} is satisfied.
\end{lem}

The following proposition states two important properties of a tempered $\alpha$-stable CBI process.

\begin{prop}	\label{prop:properties}
Let $X=(X_t)_{t\geq0}$ be a tempered $\alpha$-stable CBI process, with $X_0=x>0$, $\theta\geq0$, $\alpha\in(1,2)$ and $C=C(\alpha,\eta)$. Then the following hold:
\begin{itemize}
\item[(i)] $\EE[e^{\gamma X_t}]<+\infty$ for all $\gamma\leq\theta$ and $t\geq0$ if and only if $\phi(-\theta)\leq0$;
\item[(ii)] $0$ is an inaccessible boundary for $X$ if and only if $2\beta\geq\sigma^2$.
\end{itemize}
\end{prop}

\begin{rem}	\label{rem:complex_alpha}
In view of Remark \ref{rem:complex}-(2), part (i) of Proposition \ref{prop:properties} implies that, if $\phi(-\theta)\leq0$, then for a tempered $\alpha$-stable CBI process the affine transform formula \eqref{eq:aff_transform} can be always extended to the complex domain for all $p\in\C$ such that $\Re(p)>-\theta$.
This property will play an important role for the application of Fourier-based pricing methods (see Section \ref{sec:caplet_pricing}).
\end{rem}

\begin{rem}[The non-tempered case]	\label{rem:non_tempered}
The non-tempered case corresponds to $\theta=0$, in which case $\alpha\in(1,2)$ is required.
In this case, an $\alpha$-stable CBI process $X=(X_t)_{t\geq0}$ with $X_0=x$ can be represented as the solution to the following SDE:
\[
X_t = x + \int_0^t(\beta-bX_s)\ud s + \sigma\int_0^t\sqrt{X_s}\ud B_s + \eta\int_0^tX_{s-}^{1/\alpha}\ud Z_s,
\qquad\text{ for all }t\geq0,
\]
where $Z=(Z_t)_{t\geq0}$ is a spectrally positive $\alpha$-stable compensated L\'evy process, see \cite[Theorem 9.32]{Li}. 
This specification has been recently proposed as a model for the short interest rate in \cite{JMS17} and for stochastic volatility in \cite{JMSZ18}.
Observe that for the limit case $\alpha=2$ and $C=C(2,\eta)$, the process $X$ reduces to a classical CIR process with volatility  $\sqrt{\sigma^2+2\eta^2}$. 
This is due to the fact that, for $\theta=0$, the limit case $\alpha=2$ and $C=C(2,\eta)$ yields a Gaussian distribution for $Z$, see \cite[Section 3.7]{ContTankov}.
It is important to note that, in the non-tempered case, for any $\alpha\in(1,2)$, the process $X$ does not admit finite exponential moments of any order, meaning that $\EE[e^{\gamma X_t}]=+\infty$ for all $\gamma>0$ and $t>0$.
The finiteness of exponential moments represents an indispensable requirement of our modelling framework.
\end{rem}

\begin{rem}
In view of \cite[Proposition 4.1]{JMS17}, tempered $\alpha$-stable CBI processes are closed under a  wide class of equivalent changes of probability. More specifically, it can be shown that if $X$ is a tempered $\alpha$-stable process under $\QQ$ and $\QQ'$  is a probability measure equivalent to $\QQ$ with density process $\ud\QQ'|_{\cF_t}/\ud\QQ|_{\cF_t}=\mathcal{E}(\varrho\int_0^{\cdot}\int_0^{X_{s-}}W(\ud s,\ud u)+\int_0^{\cdot}\int_0^{+\infty}\!\int_0^{X_{s-}}(e^{-\xi z}-1)\tildeM(\ud s,\ud z,\ud u))_t$, for all $t\geq0$, for some $\varrho\in\R$ and $\xi\in\R_+$, then $X$ remains a tempered $\alpha$-stable process under $\QQ'$ (with different mean-reversion rate and tempering parameter, depending on the values of $\varrho$ and $\xi$). 
In particular, this permits to construct a tempered $\alpha$-stable CBI process from a non-tempered $\alpha$-stable CBI process by means of an equivalent change of probability.
\end{rem}

We conclude this subsection with the following result, which characterizes the ergodic distribution of a tempered $\alpha$-stable CBI process.
This result can be useful for analyzing the long-term behavior of the model.
We recall that, according to the terminology of \cite[Chapter 3]{Li}, the branching mechanism $\phi$ is said to be {\em strictly subcritical} if $b>0$.

\begin{prop}	\label{prop:ergodic}
Let $X=(X_t)_{t\geq0}$ be a tempered $\alpha$-stable CBI process, with $X_0=x$, $\theta\geq0$, $\alpha\in(1,2)$ and $C=C(\alpha,\eta)$.
If $\phi$ is strictly subcritical, then $(P_t(\cdot,x))_{t\geq0}$ converges weakly to a stationary distribution $\rho$ with Laplace transform
\be	\label{eq:ergodic_laplace}
\mathcal{L}_{\rho}(p) := \int_0^{+\infty}e^{-py}\rho(\ud y)
= \exp\left(-\beta\int_0^p\frac{z}{\phi(z)}\ud z\right),
\qquad\text{ for }p>p_0,
\ee
where $p_0$ is defined as in Theorem \ref{thm:lifetime} for $q=0$.
The first moment of $\rho$ is given by 
\be	\label{eq:ergodic_mean}
\int_0^{+\infty}y\rho(\ud y) = \frac{\beta}{b}.
\ee
Moreover, the process $X$ is exponentially ergodic, in the sense that
\[
\|P_t(\cdot,x) - \rho(\cdot)\| \leq C\bigl(x+\beta/b\bigr)e^{-bt},
\qquad\text{ for all }t\geq1,
\]
for a positive constant $C$ and where $\|\cdot\|$ denotes the total variation norm.
\end{prop}

\subsection{Model specification}	\label{sec:model_spec}

In this section, we present a tractable specification of a multi-curve model driven by a flow of CBI processes, using the class of tempered $\alpha$-stable processes introduced in Section \ref{sec:alpha_stable}.
As in Section \ref{sec:prelim_CBI}, let $(\Omega,\cF,\FF,\QQ)$ be a filtered probability space supporting a white noise $W(\ud s,\ud u)$ on $(0,+\infty)^2$ with intensity $\ud s\,\ud u$ and a Poisson time-space random measure $M(\ud s,\ud z,\ud u)$ on $(0,+\infty)^3$ with intensity $\ud s\,\pi(\ud z)\,\ud u$. 
As in Section \ref{sec:rates_spreads}, we consider a set of tenors $\cD=\{\delta_1,\ldots,\delta_m\}$. For each $i=1,\ldots,m$, let the process $Y^i=(Y^i_t)_{t\geq0}$ be the unique strong solution to the following stochastic integral equation (see Proposition \ref{prop:CBI}):
\be	\label{eq:flow}
Y^i_t = y^i_0 + \int_0^t(\beta(i)-b Y^i_{s-})\ud s + \sigma\int_0^t\int_0^{Y^i_{s-}}W(\ud s,\ud u) + \eta\int_0^t\int_0^{+\infty}\!\int_0^{Y^i_{s-}}z\tildeM(\ud s,\ud z,\ud u),
\ee
for all $t\geq0$, with initial condition $y^i_0\in\R_+$, for all $i=1,\ldots,m$, and where
\begin{itemize}
\item $\beta:\{1,\ldots,m\}\rightarrow\R_+$, with $\beta(i)\leq\beta(i+1)$, for all $i=1,\ldots,m-1$;
\item $(b,\sigma)\in\R^2$ and $\eta\geq0$;
\item $\pi(\ud z)$ is as in \eqref{eq:alpha_stable}, with $\theta>\eta$, $\alpha\in(1,2)$ and $C=C(\alpha,1)$ as in \eqref{eq:C}.
\end{itemize}

The family of processes $\{Y^i;i=1,\ldots,m\}$ is a simple instance of a {\em flow of CBI processes}, see \cite[Section 3]{DL12}. All components of the flow have a common branching mechanism $\phi$, explicitly given in Lemma \ref{lem:branch_alpha} (with the parameter $\theta$ in \eqref{eq:branch_alpha} being replaced by $\theta/\eta$, due to the appearance of $\eta$ in front of the last integral in \eqref{eq:flow}), while the immigration rate of $Y^i$ is equal to  $\psi^i(z)=\beta(i)z$, for each $i=1,\ldots,m$.
Observe that the processes $\{Y^i;i=1,\ldots,m\}$ share the same volatility coefficients $\sigma$ and $\eta$, the same jump measure $\pi$ and the same speed of mean-reversion $b$. Only the long-run value $\beta(i)/b$ (corresponding to the mean of the ergodic distribution of $Y^i$, see Proposition \ref{prop:ergodic}) is specific for each process $Y^i$, for $i=1,\ldots,m$. This modelling framework is therefore rather parsimonious on the number of model parameters. 
The fact that the martingale terms in \eqref{eq:flow} are generated by common sources of radomness $W$ and $\tildeM$ but depend on the current value of each process $Y^i$ implies a non-trivial dependence structure among the processes $Y^1,\ldots,Y^m$, as it will be shown in formula \eqref{eq:covariation} below.

In this section, we assume the validity of the following condition:
\be	\label{eq:flow_no_explosion}
b \geq \frac{\sigma^2}{2}\frac{\theta}{\eta} + \eta\frac{(1-\alpha)\theta^{\alpha-1}}{\cos(\alpha\pi/2)}.
\ee
Condition \eqref{eq:flow_no_explosion} is equivalent to $\phi(-\theta/\eta)\leq0$. In view of part (i) of Proposition \ref{prop:properties}, since $\theta>\eta$, this condition suffices to ensure that $\EE[e^{Y^i_t}]<+\infty$ for all $i=1,\ldots,m$ and $t\geq0$.
Moreover, by part (ii) of Proposition \ref{prop:properties}, condition \eqref{eq:flow_no_explosion} also ensures non-attainability of $0$.

Defining the factor process $Y=(Y_t)_{t\geq0}$ by $Y_t:=(Y^1_t,\ldots,Y^m_t)^{\top}$, for $t\geq0$, we specify the OIS short rate and spot multiplicative spreads as follows, for all $t\geq0$ and $i=1,\ldots,m$:
\begin{align}
r_t &= \ell(t) + \mu^{\top}Y_t,	\label{eq:flow_OIS}\\
\log S^{\delta_i}(t,t) &= c_i(t) + Y^i_t,	\label{eq:flow_spread}
\end{align}
where $\ell:\R_+\rightarrow\R$, with $\int_0^T|\ell(u)|\ud u<+\infty$ for all $T>0$, $c_i:\R_+\rightarrow\R_+$ and $\mu\in\R^m_+$.

Under specification \eqref{eq:flow_spread}, multiplicative spreads are by construction greater than one.
Moreover, thanks to the properties of a flow of CBI processes, monotonicity of multiplicative spreads can be easily achieved, provided that initially observed spreads are increasing in the tenor. This is the content of the following result, which relies  on a general comparison principle for CBI processes.

\begin{prop}	\label{prop:ordering}
Suppose that $y^i_0\leq y^{i+1}_0$ and $c_i(t)\leq c_{i+1}(t)$, for all $i=1,\ldots,m-1$ and $t\geq0$. 
Then it holds that $S^{\delta_i}(t,T)\leq S^{\delta_{i+1}}(t,T)$ a.s., for all $i=1,\ldots,m-1$ and $0\leq t\leq T<+\infty$.
\end{prop}
\begin{proof}
Since $\beta:\{1,\ldots,m\}\rightarrow\R_+$ is increasing, \cite[Theorem 3.2]{DL12} implies that, if $y^i_0\leq y^{i+1}_0$, then $\QQ(Y^i_t\leq Y^{i+1}_t,\text{ for all }t\geq0)=1$. Therefore, if in addition $c_i(t)\leq c_{i+1}(t)$ for all $t\geq0$, it follows that $S^{\delta_i}(t,t)\leq S^{\delta_{i+1}}(t,t)$ a.s. for all $t\geq0$. The claim follows by the fact that the process $(S^{\delta_i}(t,T))_{t\in[0,T]}$ defined in \eqref{eq:fwd_spread} is a martingale under the $T$-forward probability measure $\QQ^T$.
\end{proof}

The factors $Y^1,\ldots,Y^m$ possess the characteristic self-exciting behavior of CBI processes. This translates directly into a self-exciting property of spreads: for each $i=1,\ldots,m$, a large value of $S^{\delta_i}(t,t)$ increases the likelihood of further upward jumps of the spread itself. 
As discussed in Remark \ref{rem:jump_cluster}, a large value of $S^{\delta_i}(t,t)$ increases the volatility of the spread process itself, thereby generating volatility clustering zones in correspondence of large values of the spreads.
Under the conditions of Proposition \ref{prop:ordering}, there is a further {\em self-exciting effect among different spreads}: a large value of $S^{\delta_i}(t,t)$ increases the likelihood of upward jumps of all other spreads with tenor $\delta_j$, for $j>i$, reflecting the higher risk implicit in Ibor rates with longer tenors. 
As mentioned in Section \ref{sec:intro} (see in particular Figure \ref{fig:spreads_analysis}), these properties represent empirically relevant features of post-crisis multi-curve interest rate markets.  Figure \ref{fig:model_paths} shows a sample trajectory for a model with $\cD=\{{\rm 3M}, {\rm 6M}\}$, providing a clear evidence of jump clustering phenomena.
The sample paths have been generated by relying on the simulation method described in Appendix \ref{app:simulation}, using the parameter values reported in Table \ref{table:calibratedParameters}.

\begin{figure}[ht]
\centering
\includegraphics[scale=0.5]{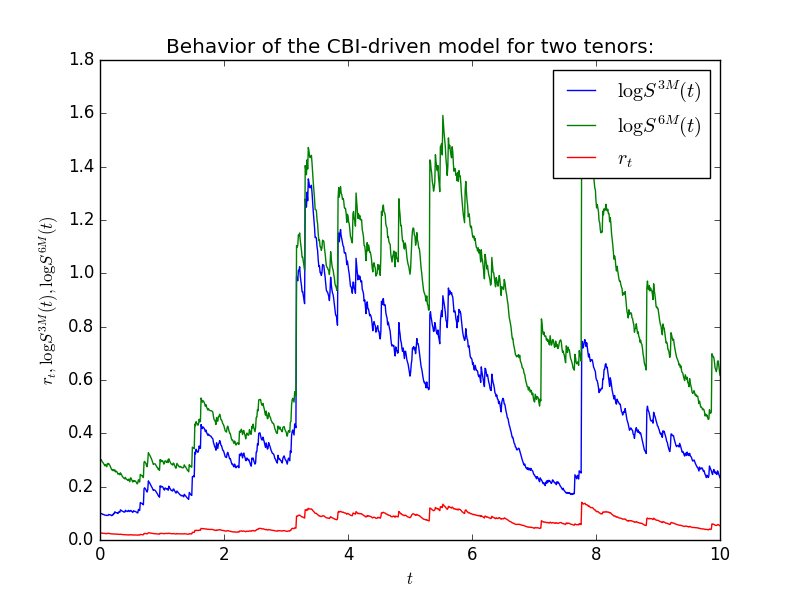}
\caption{One sample path of the short rate (red line) and multiplicative spreads for two tenors (3M in blue and 6M in green).
Parameter set reported in Table~\ref{table:calibratedParameters}.}
\label{fig:model_paths}
\end{figure}

In the model \eqref{eq:flow_OIS}-\eqref{eq:flow_spread}, each factor $Y^i$ drives the multiplicative spread with corresponding tenor $\delta_i$, while all the factors can affect the OIS short rate. This generates a non-trivial dependence between the OIS short rate and multiplicative spreads, as well as among the spreads themselves, in line with the dynamics observed on market data.
In particular, under the conditions of Proposition \ref{prop:ordering}, the quadratic covariation of log-spreads associated to tenors $\delta_i$ and $\delta_j$, with $i<j$, is given by
\be	\label{eq:covariation}
\bigl[\log S^{\delta_i}(\cdot,\cdot), \log S^{\delta_j}(\cdot,\cdot)\bigr]
= \sigma^2 \int_0^{\cdot}Y^i_s\ud s + \eta^2\int_0^{\cdot}\int_0^{+\infty}\int_0^{Y^i_{s-}}z^2M(\ud s,\ud z,\ud u).
\ee
This representation of the quadratic covariation between log-spreads shows that common jumps arise due to the presence of the common random measure $M(\ud s,\ud z,\ud u)$. The feature of common upward jumps is consistent with the empirical evidence reported in the introduction and is clearly visible in the simulated paths of Figure \ref{fig:model_paths}.

The components of the flow of CBI processes $\{Y^i;1,\ldots,m\}$ are highly dependent. Therefore, model \eqref{eq:flow_OIS}-\eqref{eq:flow_spread}  in its present form does not belong to the class of CBI-driven multi-curve models as introduced in Definition \ref{def:general_model}. However, an easy transformation allows embedding the present specification into the framework of Section \ref{sec:framework}.
To this effect, we define the $m$-dimensional process $X=(X_t)_{t\geq0}$ by
\be	\label{eq:process_X}
X^i_t := Y^i_t - Y^{i-1}_t,
\qquad\text{ for all }t\geq0\text{ and }i=1,\ldots,m,
\ee
with $Y^0\equiv 0$, and $X_t:=(X^1_t,\ldots,X^m_t)^{\top}$. 
We are now in a position to state the following result.

\begin{prop}	\label{prop:flow_CBI}
Suppose that $y^i_0\leq y^{i+1}_0$, for all $i=1,\ldots,m-1$. Let us define $\lambda\in\R^m_+$ and $\boldsymbol{\gamma}=(\gamma_1,\ldots,\gamma_m)\in\R^{m\times m}_+$ by
\be	\label{eq:def_lambda_gamma}
\lambda_j := \sum_{k=j}^m\mu_k
\qquad\text{and}\qquad
\gamma_{i,j} := \ind_{\{j\leq i\}},
\qquad\text{ for all }i,j=1,\ldots,m.
\ee
Then the tuple $(X,\ell,\lambda,\mathbf{c},\boldsymbol{\gamma})$ generates a CBI-driven multi-curve model such that
\begin{enumerate}
\item[(i)] for each $i=1,\ldots,m$, the process $X^i=(X^i_t)_{t\geq0}$ is a tempered $\alpha$-stable CBI process with branching mechanism $\phi$ and immigration rate $\psi^i(z)=(\beta(i)-\beta(i-1))z$, where $\beta(0):=0$;
\item[(ii)] the processes $(X^1,\ldots,X^m)$ are mutually independent;
\item[(iii)] the OIS short rate and multiplicative spreads are given by \eqref{eq:flow_OIS}-\eqref{eq:flow_spread}.
\end{enumerate}
\end{prop}
\begin{proof}
Parts (i) and (ii) are a direct consequence of \cite[Theorems 3.2 and 3.3]{DL12}. To prove part (iii), it suffices to observe that, due to the definition of $\lambda$ and $\boldsymbol{\gamma}$ in \eqref{eq:def_lambda_gamma}, it holds that
\[
\mu^{\top}Y_t = \lambda^{\top}X_t
\qquad\text{and}\qquad
Y^i_t = \gamma_i^{\top}X_t,
\qquad\text{ for all }i=1,\ldots,m\text{ and }t\geq0.
\]
Note that condition \eqref{eq:general_model_life} is implied by condition \eqref{eq:flow_no_explosion}, since $\cY^i=[-\theta/\eta,+\infty)$ (see Lemma \ref{lem:branch_alpha}), for all $i=1,\ldots,m$, where $\theta>\eta$, and due to part (i) of Proposition \ref{prop:properties}.
\end{proof}

In view of the above proposition, the model \eqref{eq:flow_OIS}-\eqref{eq:flow_spread} driven by the CBI flow $\{Y^i;i=1,\ldots,m\}$ can be equivalently represented by a family of independent risk factors $(X^1,\ldots,X^m)$, where each factor $X^i$ is affecting all spreads with tenor $\delta_j\geq\delta_i$ (and, possibly, the OIS  rate). 
In particular, the presence of common risk factors among different spreads accounts for the possibility of common upward jumps, as mentioned above, in line with the stylized facts reported in Section \ref{sec:intro}.

OIS zero-coupon bond prices and forward multiplicative spreads  can be explicitly computed by relying on Proposition \ref{prop:bond_spreads}. 
More specifically, it holds that, for all $i=1,\ldots,m$,
\be	\label{eq:AB}	\begin{aligned}
\cA_0(t,T) &= -\int_0^{T-t}\ell(t+s)\ud s - \sum_{j=1}^m\bigl(\beta(j)-\beta(j-1)\bigr)\int_0^{T-t}v\biggl(s,0,\sum_{k=j}^m\mu_k\biggr)\ud s,	\\
\cB^j_0(T-t) &= -v\biggl(T-t,0,\sum_{k=j}^m\mu_k\biggr),
\qquad\text{ for all }j=1,\ldots,m\\
\cA_i(t,T) &= c_i(T) + \sum_{j=1}^i\left(\bigl(\beta(j)-\beta(j-1)\bigr)\int_0^{T-t}\Biggl(v\biggl(s,0,\sum_{k=j}^m\mu_k\biggr)-v\biggl(s,-1,\sum_{k=j}^m\mu_k\biggr)\Biggr)\ud s\right),\\
\cB^j_i(T-t) &= \Biggl(v\biggl(T-t,0,\sum_{k=j}^m\mu_k\biggr) - v\biggl(T-t,-1,\sum_{k=j}^m\mu_k\biggr)\Biggr)\ind_{\{j\leq i\}},
\qquad\text{ for all }j=1,\ldots,m.
\end{aligned}	\ee
Observe that, unlike in the general framework of Section \ref{sec:framework}, the function $v$ appearing in the above formulae is the same for all $i,j=1,\ldots,m$, due to the fact that the components of a flow of CBI processes share a common branching mechanism $\phi$. 
This results in additional analytical tractability of the  present modelling framework.
We recall that $v$ is given by the unique solution to the ODE \eqref{eq:ODE_v} with $\phi$ given as in \eqref{eq:branch_alpha}, with the parameter $\theta$ replaced by $\theta/\eta$.

\subsection{Valuation of caplets via Fourier methods}	\label{sec:caplet_pricing}

In this section, we provide a semi-closed formula for the price of a caplet.
Let us consider a caplet written on the Ibor rate with tenor $\delta_i$, strike $K>0$, maturity $T>0$ and settled in arrears at time $T+\delta_i$. For simplicity of presentation, we consider a unitary notional amount. 
By formula (A.1) in \cite{CFGaffine}, the arbitrage-free price of the caplet can be expressed as
\begin{align}	\Pi^{{\rm CPLT}}(t;T,\delta_i,K,1) 
&= \delta_i\,\EE\left[e^{-\int_t^{T+\delta_i}r_s\ud s}\bigl(L(T,T,\delta_i)-K\bigr)^+\bigr|\cF_t\right]	\notag\\
&= B(t,T+\delta_i)\EE^{T+\delta_i}\left[\bigl(e^{\cX_T^i}-(1+\delta_i K)\bigr)^+\Bigr|\cF_t\right],
\qquad\text{ for }t\leq T,
\label{eq:gen_caplet}
\end{align}
where 
the process $\cX^i=(\cX_t^i)_{t\geq0}$ is defined by $\cX_t^i:=\log(S^{\delta_i}(t,t)/B(t,t+\delta_i))$, for all $t\geq0$.
As a consequence of Proposition \ref{prop:bond_spreads}, \eqref{eq:flow_spread} and Proposition \ref{prop:flow_CBI}, the process $\cX^i$ admits the explicit representation
\[
\cX^i_t = c_i(t) - \cA_0(t,t+\delta_i) + \bigl(\gamma_i-\cB_0(\delta_i)\bigr)^{\top}X_t,
\qquad\text{ for all }t\geq0,
\]
where the $d$-dimensional process $(X_t)_{t\geq0}$ is defined in \eqref{eq:process_X} and the functions $\cA_0$ and $\cB_0$ in \eqref{eq:AB}.
For $T\geq0$ and $i=1,\ldots,m$, let us introduce the set
\[
\Theta_i(T) := \bigl\{u\in\R : \EE^{T+\delta_i}[e^{u\cX_T^i}]<+\infty\bigr\}^{\circ}
\]
and the strip $\Lambda_i(T):=\{\zeta\in\C : -\Im(\zeta)\in\Theta_i(T)\}$. Using condition \eqref{eq:flow_no_explosion} and the fact that $-\cB^j_0(\delta_i)=v(\delta_i,0,\lambda_j)\geq0$, for all $i,j=1,\ldots,m$, together with the increasingness of  the map $\R_+\ni q\mapsto v(\delta_i,0,q)$, it can be checked that  the condition
\be	\label{eq:int_fourier}
u< \frac{\theta/\eta+v(\delta_i,0,\lambda_1)}{1+v(\delta_i,0,\lambda_1)}
\ee
is sufficient to ensure that $u\in\Theta_i(T) $, for all $T>0$.
In particular, it always holds that $(-\infty,+1]\subseteq\Theta_i(T)$.
For $\zeta\in\Lambda_i(T)$, the {\em modified characteristic function} of $\cX^i_T$ can be defined and explicitly computed as follows:
\be	\label{eq:mod_char_fun}	\begin{aligned}
\Phi^i_{t,T}(\zeta)
&:= B(t,T+\delta_i)\,\EE^{T+\delta_i}\bigl[e^{\im\zeta\cX_T^i}|\cF_t\bigr]
= \EE\left[e^{-\int_t^Tr_s\ud s}B(T,T+\delta_i)e^{\im\zeta\cX_T^i}\Bigr|\cF_t\right]	\\
&=  \exp\left( - \int_0^{T-t}\ell(t+s)\ud s
+ (1-\im\zeta)\cA_0(T,T+\delta_i)+\im\zeta c_i(T)\right.\\
&\qquad\quad\;\;\left.	
-\sum_{j=1}^m\bigl(\beta(j)-\beta(j-1)\bigr)\int_0^{T-t}v\bigl(s,(\im\zeta-1)\cB_0^j(\delta_i)-\im\zeta\gamma_{i,j},\lambda_j\bigr)\ud s\right.\\
&\qquad\quad\;\;\left.
- \sum_{j=1}^mv\bigl(T-t,(\im\zeta-1)\cB_0^j(\delta_i)-\im\zeta\gamma_{i,j},\lambda_j\bigr)X^j_t\right).
\end{aligned}	\ee
We remark that, under condition \eqref{eq:flow_no_explosion}, the above application of the affine transform formula \eqref{eq:aff_transform} in the complex domain is justified by Remark \ref{rem:complex_alpha}. 
More specifically, $\zeta\in\Lambda_i(T)$ ensures that $\Re((\im\zeta-1)\cB^j_0(\delta_i)-\im\zeta\gamma_{i,j})>-\theta/\eta$. Therefore, by condition \eqref{eq:flow_no_explosion} together with Remark \ref{rem:complex_alpha}, it follows that $v(t,\Re((\im\zeta-1)\cB^j_0(\delta_i)-\im\zeta\gamma_{i,j}),\lambda_j)>-\theta/\eta$, for all $t\geq0$, meaning that the solution to the ODE \eqref{eq:ODE_v} with $p=\Re((\im\zeta-1)\cB^j_0(\delta_i)-\im\zeta\gamma_{i,j})$ stays in $\cY^{\circ}$ (compare with Remark \ref{rem:complex}-(2)).

We are now in a position to state the caplet valuation formula, which is a direct consequence of \cite[Theorem 5.1]{lee2004}. Note that, according to the notation of \cite{lee2004}, in our case we have that $G=G_1$ and $b_0=b_1=1\in\Theta_i(T)$ (due to the fact that $\theta>\eta$ and condition \eqref{eq:flow_no_explosion} holds).

\begin{prop}	\label{prop:caplet}
Let $\bar{K}_i:=1+\delta_i K$ and $\epsilon\in\R$ such that $1+\epsilon\in\Theta_i(T)$. The arbitrage-free price at time $t\leq T$ of a caplet written on the Ibor rate with tenor $\delta_i$, strike $K>0$, maturity $T>0$ and settled in arrears at time $T+\delta_i$, is given by
\[
\Pi^{{\rm CPLT}}(t;T,\delta_i,K,1)
= R^i_{t,T}\left(\bar{K}_i,\epsilon\right)+\frac{1}{\pi }\int_{0-\im\epsilon}^{\infty-\im\epsilon}\Re\left(e^{-\im\zeta \log(\bar{K}_i)}\frac{\Phi^i_{t,T}(\zeta-\im)}{-\zeta(\zeta-\im)}\right)\ud\zeta,  
\]
where $\Phi^i_{t,T}$ is given in \eqref{eq:mod_char_fun} and $R^i_{t,T}(\bar{K}_i,\epsilon)$ is given by
\[
R^i_{t,T}\left(\bar{K}_i,\epsilon\right) = 
\begin{cases}
 \Phi^i_{t,T}(-\im)-\bar{K}_i\Phi^i_{t,T}(0), & \mbox{if } \epsilon<-1, \\ 
 \Phi^i_{t,T}(-\im)-\frac{\bar{K}_i}{2}\Phi^i_{t,T}(0), & \mbox{if } \epsilon=-1, \\
   \Phi^i_{t,T}(-\im), & \mbox{if } -1<\epsilon<0,\\
    \frac{1}{2}\Phi^i_{t,T}(-\im),&\mbox{if } \epsilon = 0,\\
    0,&\mbox{if } \epsilon >0.
 \end{cases}
\]
\end{prop}

\subsection{Valuation of caplets via quantization}	\label{sec:caplet_pricing_quant}

The analytical tractability of CBI processes allows for the development of a quantization-based pricing methodology, which is here proposed for the first time in an interest rate model. 
In this section, we show that the Fourier-based quantization technique recently introduced in \cite{CFG18} can be easily applied for the pricing of caplets.

The key ingredient of this approach is represented by the {\em quantization grid} $\Gamma^N=\{x_1,\ldots,x_N\}$, with $x_1<\ldots<x_N$, for some chosen $N\in\N$ (see \cite{gl2000,g19} for details).
Once the quantization grid $\Gamma^N$ has been determined, the random variable $e^{\cX_T^i}$ appearing in the caplet valuation formula \eqref{eq:gen_caplet} is approximated by its {\em Voronoi $\Gamma^N$-quantization}, i.e., the nearest neighbour projection $\widehat{e^{\cX_T^i}}$ of $e^{\cX_T^i}$ onto $\Gamma^N$, given by the discrete random variable
\[
\widehat{e^{\cX_T^i}} = \sum_{j=1}^{N} x_{j} \ind_{\{x_{j}^{-} \leq e^{\cX_T^i} \leq x_{j}^{+}\}},
\]
where $x_{j}^{-} = (x_{j-1}+x_{j})/2$ and $x_{j}^{+} = (x_{j+1}+x_{j})/2$,  for $j=1,\ldots,N$, with $x_{1}^{-} = 0$ and $x_{N}^{+} = +\infty$. 
Formula \eqref{eq:gen_caplet} can then be approximated as follows (considering $t=0$ for simplicity of presentation):
\[
\Pi^{{\rm CPLT}}(0;T,\delta_i,K,1) 
\approx B(0,T+\delta_i)\sum_{j=1}^{N} \bigl(x_{j}-(1+K\delta_{i})\bigr)^{+}\mathbb{Q}^{T+\delta_i}\bigl(\widehat{e^{\cX_T^i}} = x_{j}\bigr),
\]
where the {\em companion weights} $\mathbb{Q}^{T+\delta_i}(\widehat{e^{\cX_T^i}} = x_{j})$, for $j=1,\ldots,N$, are computed by
\be \label{eq:companion_weights}
\mathbb{Q}^{T+\delta_i}\bigl(\widehat{e^{\cX_T^i}} = x_{j}\bigr) 
= \mathbb{Q}^{T+\delta_i}\bigl(e^{\cX_T^i} \leq x_{j}^{+}\bigr) - \mathbb{Q}^{T+\delta_i}\bigl(e^{\cX_T^i} \leq x_{j}^{-}\bigr).  
\ee

The core of quantization consists in optimally determining the quantization grid $\Gamma^N$ in such a way that the discrete distribution of $\widehat{e^{\cX_T^i}}$ over $\Gamma^{N}$ is a good approximation of the continuous distribution of $e^{\cX_T^i}$. This is achieved by choosing a grid $\Gamma$ that minimizes the following $L^p$-distance:
\be \label{eq:distortion_function}
D_{p}(\Gamma) = D_p(\{x_1,\ldots,x_N\}) 
:= \bigl\|e^{\cX_T^i}-\widehat{e^{\cX_T^i}}\bigr\|_{L^p(\mathbb{Q}^{T+\delta_i})} \\
= \mathbb{E}^{T+\delta_i}\left[\min_{j=1,\ldots,N} \bigl|e^{\cX_T^i}-x_{j}\bigr|^{p}\right]^{1/p}.
\ee
In the present setting, it can be shown that this minimization problem admits a unique solution of full size $N$ (see \cite[Proposition 1.1]{g19}). In practice, $\Gamma^N$ is typically determined by searching the critical points of the map $\Gamma\mapsto D_p(\Gamma)$ ({\em sub-optimal} quantization grid). In view of \cite[Theorem 1]{CFG18}, a sub-optimal quantization grid $\Gamma^N=\{x_1,\ldots,x_N\}$ is given by the solution to the following equation:\begin{equation}	\label{eq:master_equation} 
\int_{0}^{+\infty}{ \Re\left[\Psi_{T}^{i}(u)e^{-\im u\log(x_j)}\left(\bar{\beta}\left(\frac{x_j^-}{x_j},-\im u,p\right)-\bar{\beta}\left(\frac{x_j}{x_j^+},1-p+\im u,p\right)\right)\right] \ud u} = 0, \end{equation}
for $j=1,\ldots,N$, where $\bar{\beta}$ is defined as
\begin{equation*}
\bar{\beta}(x,a,b) = \int_{x}^{1} {t^{a-1}(1-t)^{b-1}\ud t},
\qquad\text{ for $a \in \mathbb{C}$, $\Re(b) > 0$ and $x \in (0,1)$},
\end{equation*}
and $\Psi_{T}^{i}$ stands for the $(T+\delta_i)$-forward characteristic function of  $\cX_T^i$:
\[
\Psi_{T}^{i}(u) := \mathbb{E}^{T+\delta_i}[e^{\im u \cX_T^i}] = \frac{\Phi^i_{0,T}(u)}{B(0,T+\delta_i)}.
\]

Equation \eqref{eq:master_equation} can be efficiently solved by relying on Newton-Raphson-type algorithms. Indeed, in the present framework, the gradient $\nabla D_p$ of the function $D_p$ can be analytically computed and the associated Hessian matrix $H[D_p]$ turns out to be tridiagonal. 
To initialize the algorithm, the starting grid $\Gamma^N_{(0)}$ can be constructed by using a regular spacing around the expectation of the state variable $e^{\cX_T^i}$, which is directly determined by market observables:
\[
\mathbb{E}^{T+\delta_i}\bigl[e^{\cX_T^i}\bigr] 
= \mathbb{E}^{T+\delta_i}\left[\frac{S^{\delta_{i}}(T,T)}{B(T,T+\delta_i)}\right] 
= 1 + \delta_{i}L(0,T,\delta_i).
\]
Starting from $\Gamma^N_{(0)}$, a basic formulation of the Newton-Raphson algorithm for the determination of a sub-optimal quantization grid $\Gamma^N$ is then based on the following iterations:
\[
\Gamma_{(n+1)}^{N} = \Gamma_{(n)}^{N} - \left(H[D_{p}](\Gamma_{(n)}^{N})\right)^{-1}\nabla D_{p}(\Gamma_{(n)}^{N}), \quad \text{at each iteration $n\in\N$.}
\]

\begin{rem}
We emphasize that the companion weights $\mathbb{Q}^{T+\delta_i}(\widehat{e^{\cX_T^i}}=x_j)$, for $j=1,\ldots,N$, and the density function of the random variable $e^{\cX_T^i}$ needed for the computation of the function $D_p(\Gamma)$ in \eqref{eq:distortion_function} can be recovered from the $(T+\delta_i)$-forward characteristic function $\Psi^i_T$. Indeed, it holds that
\begin{eqnarray*}
	\mathbb{Q}^{T+\delta_i}\bigl(e^{\cX_T^i} \in dx\bigr) &=& \left(\frac{1}{x\pi}\int_{0}^{+\infty}{\Re\left(e^{-\im u \log(x)}\Psi_{T}^{i}(u)\right)du}\right)\ud x, \\
	\mathbb{Q}^{T+\delta_i}\bigl(e^{\cX_T^i} \leq x\bigr) &=& \frac{1}{2} - \frac{1}{\pi}\int_{0}^{+\infty}{\Re\left(\frac{e^{-\im u \log(x)}\Psi_{T}^{i}(u)}{\im u}\right)\ud u}.
\end{eqnarray*}
As shown in \eqref{eq:mod_char_fun}, $\Psi^i_T$ can be analytically computed for a CBI-driven multi-curve model.
\end{rem}

\section{Numerical Results and Model Calibration}\label{sec:calibration}

In this section, we present some numerical results. In particular, we calibrate the model introduced in Section \ref{sec:flow} to market data relative to the 3M and 6M tenors.

\subsection{Numerical comparison of pricing methods}

We implement the Fourier and the quanti\-zation-based pricing methodologies developed in Sections \ref{sec:caplet_pricing}-\ref{sec:caplet_pricing_quant}\footnote{The Java language has been used for the whole calibration procedure. The source code is available on the website \url{https://github.com/AlessandroGnoatto/CBIMultiCurve}.}. 
In order to assess the reliability of both approaches, we compare them under different combinations of moneyness, maturities and model parameters.
We preliminarily validate the FFT methodology by means of Monte Carlo simulations and, by relying on the simulation method described in Appendix \ref{app:simulation}, we verify that caplet prices computed by FFT correspond to Monte Carlo prices. This validation procedure enables us to take FFT prices as benchmark in the following.
We then compare the FFT and quantization pricing methods. Table~\ref{table:priceComparison} shows the results of this comparison, reporting the percentage difference between FFT and quantization prices for caplets with strikes $1\%$ and $2\%$ and maturities ranging from $1$ up to $2$ years. 
This comparison over different strikes and maturities enables us to
evaluate the reliability of the quantization approach against the  FFT methodology. In Table~\ref{table:priceComparison}, we use an FFT with 4096 points and a quantization grid of 10 points. The two methodologies have different computation times. For the parameter set considered in Table~\ref{table:calibratedParameters}, FFT prices are obtained in 2 seconds for each maturity.
On the other hand, the computation time for quantization is initially lower (0.5 seconds) but then, as the maturity increases,  quantization becomes computationally more expensive, with an average computation time of 3 seconds for larger maturities.

As a further example, we report in Table~\ref{table:priceComparison2} a comparison between the two pricing methodologies for a different parameter set, corresponding to increased volatility. More specifically, we increase the parameters $\sigma$ and $\eta$ by $50\%$ with respect to the calibrated values reported in Table~\ref{table:calibratedParameters} and we set $\alpha=1.8$. The first run of the comparison was not totally satisfactory since we observed that the prices produced by quantization (using a grid with 10 points) were diverging from those obtained via FFT. This issue has been solved by increasing the number of points for the quantization from $10$ to $20$, leading to an accuracy comparable to Table~\ref{table:priceComparison}. This analysis shows that some care is needed when applying the quantization pricing methodology:  for some parameter set one needs to re-adjust the meta-parameters of the numerical scheme, a task which is difficult to perform during a calibration procedure. In summary, we deem the FFT approach more robust as well as more computationally efficient in view of the calibration of the model.

\begin{table}[ht]
\begin{tabular}{|c|c|c|c|c|c|c|}
\hline
& FFT $2\%$ & Quant. $2\%$ &Difference $2\%$ & FFT $1\%$ & Quant. $1\%$ &  Difference $1\%$ \\
\hline
$1$&$0.0064146$ & $0.0063414$ & $1.1554\%$ & $0.0065360$ & $0.0070082$ & $-6.7368\%$ \\
$1.1$&$0.0064303$ & $0.0063641$ & $1.0404\%$ & $0.0065520$ & $0.0070325$ & $-6.8330\%$ \\
$1.2$&$0.0064460$ & $0.0063868$ & $0.92556\%$ & $0.0065679$ & $0.0070569$ & $-6.9292\%$ \\
$1.3$&$0.0064679$ & $0.0064187$ & $0.76738\%$ & $0.0065903$ & $0.0070910$ & $-7.0616\%$ \\
$1.4$&$0.0064983$ & $0.0064626$ & $0.55350\%$ & $0.0066212$ & $0.0071381$ & $-7.2408\%$ \\
$1.5$&$0.0065291$ & $0.0065069$ & $0.34038\%$ & $0.0066525$ & $0.0071857$ & $-7.4193\%$ \\
$1.6$&$0.0065600$ & $0.0065516$ & $0.12813\%$ & $0.0066839$ & $0.0072335$ & $-7.5972\%$ \\
$1.7$&$0.0065911$ & $0.0065966$ & $-0.083181\%$ & $0.0067157$ & $0.0072818$ & $-7.7743\%$ \\
$1.8$&$0.0066250$ & $0.0066458$ & $-0.31232\%$ & $0.0067502$ & $0.0073345$ & $-7.9664\%$ \\
$1.9$&$0.0066620$ & $0.0070707$ & $-5.7812\%$ & $0.0067878$ & $0.0078016$ & $-12.995\%$ \\
$2$&$0.0066992$ & $0.0071339$ & $-6.0934\%$ & $0.0068257$ & $0.0078694$ & $-13.264\%$ \\  
\hline 
\end{tabular}
\caption{Comparison of FFT and quantization prices for different maturities (strikes at $2\%$ and $1\%$,  differences in relative terms). Quantization with $10$ points and FFT with $4096$ points. The parameter set used here is reported in Table \ref{table:calibratedParameters}.
\label{table:priceComparison}}
\end{table}

\begin{table}[ht]
\begin{tabular}{|c|c|c|c|c|c|c|}
\hline
& FFT $2\%$ & Quant. $2\%$ &Difference $2\%$ & FFT $1\%$ & Quant. $1\%$ &  Difference $1\%$ \\
\hline
$1$&$0.0064262$ & $0.0063429$ & $1.3127\%$ & $0.0065478$ & $0.0070095$ & $-6.5867\%$ \\
$1.1$&$0.0064433$ & $0.0063672$ & $1.1951\%$ & $0.0065652$ & $0.0070356$ & $-6.6849\%$ \\
$1.2$&$0.0064605$ & $0.0063917$ & $1.0773\%$ & $0.0065827$ & $0.0070618$ & $-6.7832\%$ \\
$1.3$&$0.0064841$ & $0.0064252$ & $0.91574\%$ & $0.0066067$ & $0.0070978$ & $-6.9182\%$ \\
$1.4$&$0.0065162$ & $0.0064710$ & $0.69830\%$ & $0.0066394$ & $0.0071468$ & $-7.1000\%$ \\
$1.5$&$0.0065487$ & $0.0065173$ & $0.48133\%$ & $0.0066725$ & $0.0071965$ & $-7.2815\%$ \\
$1.6$&$0.0065814$ & $0.0065641$ & $0.26495\%$ & $0.0067058$ & $0.0072466$ & $-7.4625\%$ \\
$1.7$&$0.0066145$ & $0.0066113$ & $0.049213\%$ & $0.0067395$ & $0.0072973$ & $-7.6430\%$ \\
$1.8$&$0.0066505$ & $0.0066628$ & $-0.18451\%$ & $0.0067761$ & $0.0073525$ & $-7.8387\%$ \\
$1.9$&$0.0066896$ & $0.0070315$ & $-4.8632\%$ & $0.0068159$ & $0.0077575$ & $-12.138\%$ \\
$2$&$0.0067291$ & $0.0070963$ & $-5.1748\%$ & $0.0068561$ & $0.0078270$ & $-12.405\%$   \\
\hline 
\end{tabular}
\caption{Comparison of FFT and quantization prices for different maturities (strikes at $2\%$ and $1\%$,  differences in relative terms). Quantization with $20$ points and FFT with $4096$ points. Starting from the parameter set of Table \ref{table:calibratedParameters}, we increased $\sigma$ and $\eta$ by $50\%$ and set $\alpha = 1.8$.
\label{table:priceComparison2}}
\end{table}

\subsection{Model calibration} 

To illustrate the calibration of our model to market data, we start by describing the market data and the reconstruction of the term structures.

\subsubsection{Market data}

We consider market data for the EUR market as of 25 June 2018, consisting of both linear and non-linear interest rate derivatives. The set of tenors is $\mathcal{D}=\left\{{\rm 3M, 6M}\right\}$. Market data on linear products consist of OIS and interest rate swaps, from which the discount curve $T \mapsto B(0, T)$ and the forward curves $T \mapsto L\left(0,T, \delta_{i}\right)$, for $\delta_{1}={\rm 3M}$ and  $\delta_{2}={\rm 6M}$, are constructed using the bootstrapping procedure from the Finmath Java library (see \cite{finmath,fries12}). 
The OIS discount curve is bootstrapped from OIS swaps, using cubic spline interpolation on logarithmic discount factors with constant extrapolation. Similarly, the 3M and 6M forward curves are bootstrapped from market quotes of FRAs (for short maturities) and swaps (for maturities beyond two years), using cubic spline interpolation on forwards with constant extrapolation. Figure~\ref{fig:curves} reports the resulting discount and forward curves. We can observe that the spread between the 3M and the 6M curves is more pronounced below twelve years and decreases afterwards. We also notice that, for short maturities, discount factors are larger than one and forward rates are negative.

\begin{figure}[ht]
  \centering
  \subfloat{\label{fig:discountCurve}\includegraphics[scale=0.4]{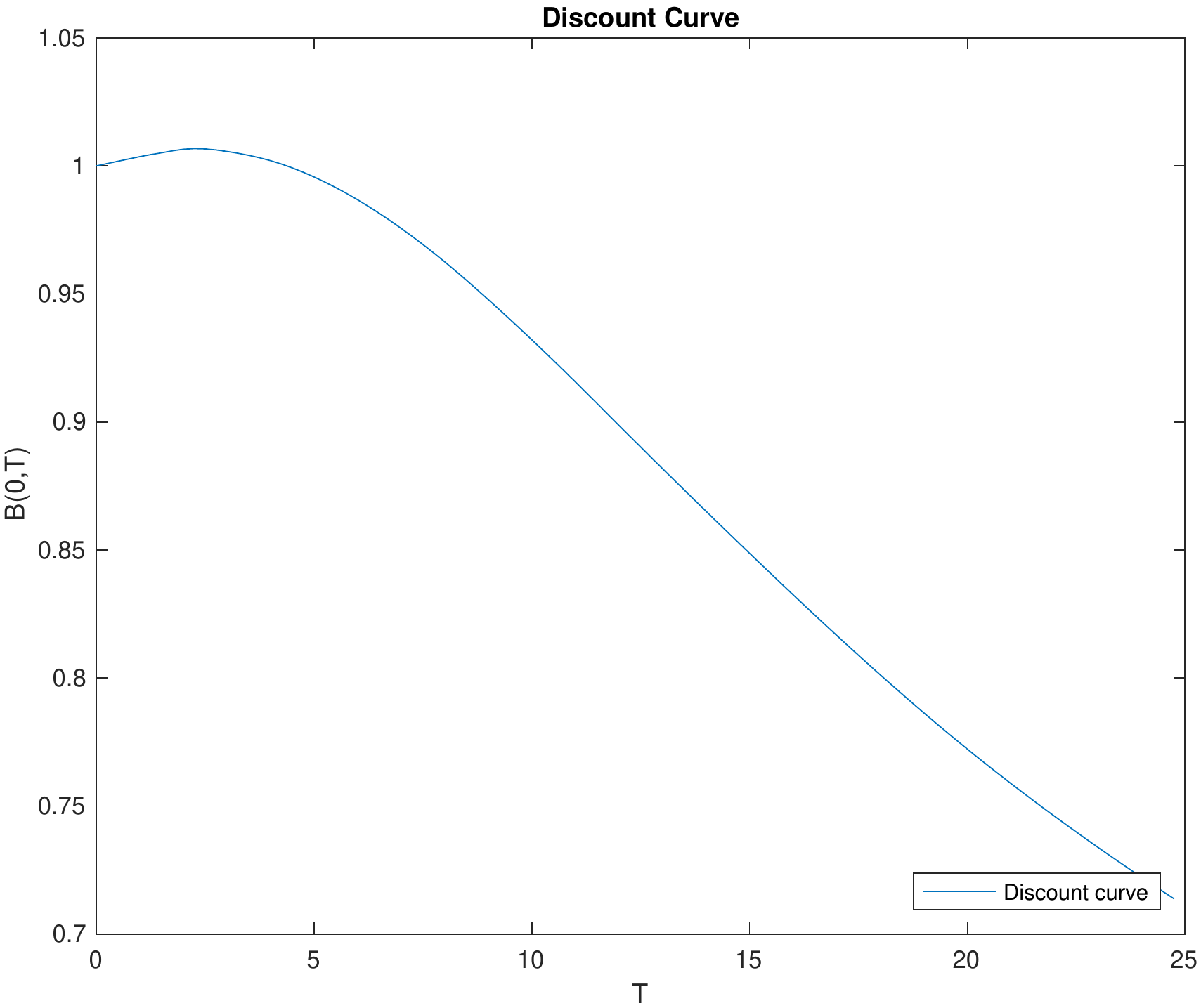}} \                                          \subfloat{\label{fig:forwardCurve}\includegraphics[scale=0.4]{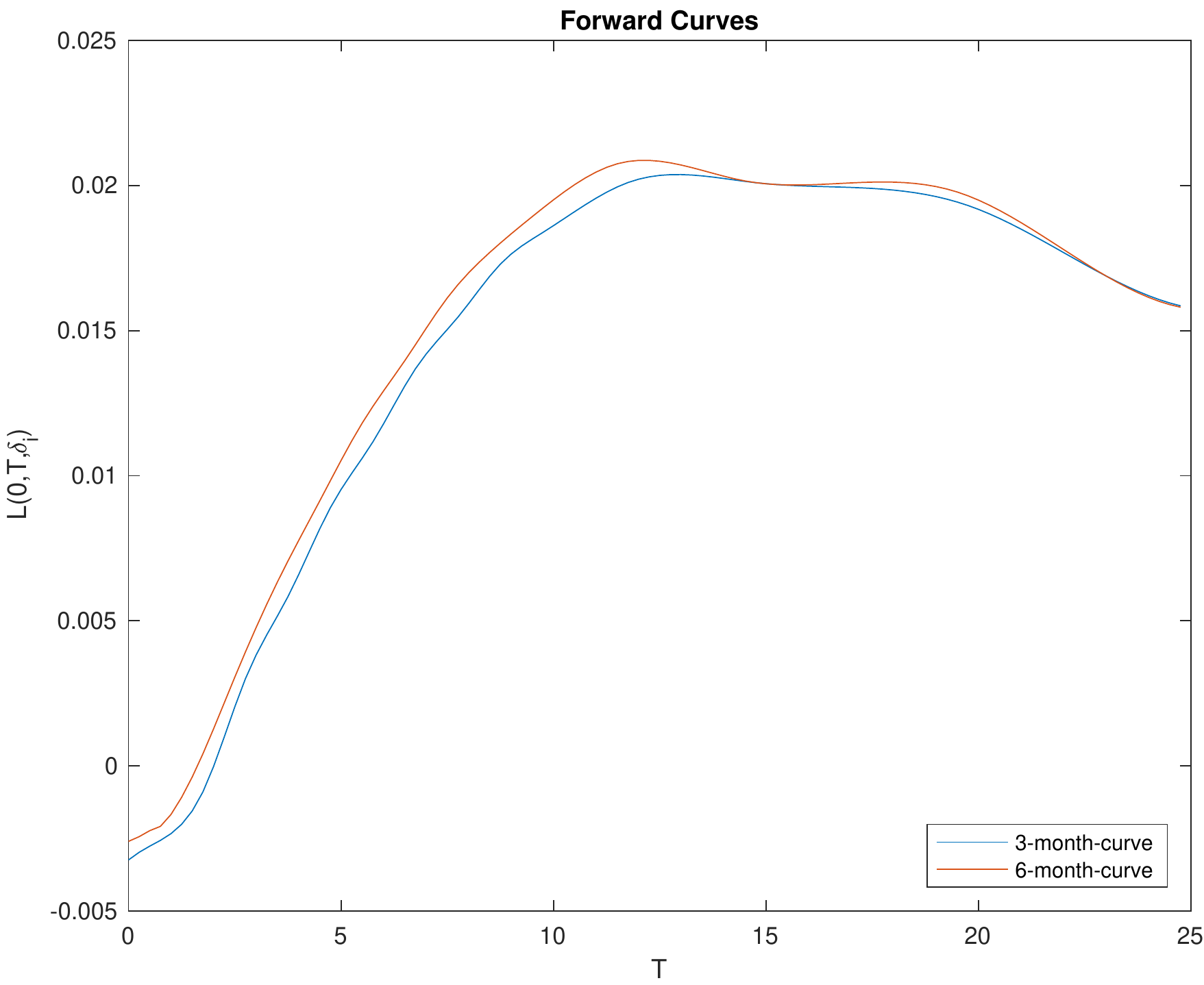}}
\caption{Discount and forward curves as of 25 June 2018. \label{fig:curves}}
\end{figure}

Concerning non-linear interest rate products, we focus on caplet market data, suitably bootstrapped from market cap volatilities. Consistently with the presence of negative interest rates, we also have market quotes for caps having a negative strike rate. Therefore, the boostrapped caplet volatility surface refers to strike prices ranging between $-0.13\%$ and $2\%$ and  maturities between $6$ months and $6$ years. Caplets with maturity larger than two years are indexed to the $6$-month rate while those with shorter expiry are linked to the $3$-month curve. Market data are given in terms of {\em normal} implied volatilities.
A normal implied volatility is obtained by numerically searching for the value of $\sigma_{\rm mkt}^{\rm imp}(K,T_i)$ such that the Bachelier pricing formula for a caplet
\be
\label{eq:caplet_bachelier}
\Pi^{\rm CPLT}_{Bac}(t; T_{i-1},\delta_{i},K,1)
:=B(t,T_i)\delta\sigma_{\rm mkt}^{\rm imp}(K,T_i)\sqrt{T_{i-1}-t}\left(\frac{1}{\sqrt{2\pi}}e^{-\frac{z^2}{2}}+zN(z)\right),
\ee
with
\[
N(x)=\frac{1}{\sqrt{2\pi}}\int_{-\infty}^xe^{-\frac{y^2}{2}}\ud y
\qquad\text{and}\qquad
z = \frac{L(t,T_{i-1},\delta_{i})-K}{\sigma_N\sqrt{T_{i-1}-t}},
\]
provides the best fit to the market price of a given caplet.

\subsubsection{Implementation} 

For a vector $p$ of model parameters belonging to the set $\mathcal{P}$ of admissible values (see Section \ref{sec:model_spec}), we compute model-implied caplet prices by means of the Fourier approach of Section \ref{sec:caplet_pricing} (see Proposition \ref{prop:caplet}). The numerical integration is performed by means of the FFT approach of \cite{cm99}, with $32768$ points and integration mesh size $0.05$. For a fixed maturity, a single execution of the FFT routine yields a vector of model prices for several moneyness levels. Prices are then converted into normal implied volatilities by using formula \eqref{eq:caplet_bachelier}. Repeating this procedure for different maturities, we generate a corresponding model-implied volatility $\sigma_{\rm mod}^{\rm imp}(K_k,T_j,p)$ for each strike $K_k$ and maturity $T_j$ present in our sample of market data.

The aim of the calibration procedure is to find the vector of parameters which solves the problem
\be	\label{eq:calibration_problem}
\min_{p\in\mathcal{P}}\sum_{j,k}\left(\sigma_{\rm mkt}^{\rm imp}(K_k,T_j)-\sigma_{\rm mod}^{\rm imp}(K_k,T_j,p)\right)^{2}.
\ee

\subsubsection{Calibration results} 
\label{sec:calibration_results}
We calibrated a two-dimensional version of the model of Section~\ref{sec:model_spec}. To solve problem \eqref{eq:calibration_problem}, we used the multi-threaded Levenberg-Marquardt optimizer of the Finmath Java library with 8 threads, imposing the parameter restrictions given after equation \eqref{eq:flow}.
The calibrated values of the parameters are reported in Table~\ref{table:calibratedParameters}. 
We can observe that the calibration results demonstrate an important role of the jump component, apparently more important than the diffusive component. Moreover, the calibrated value of $\alpha$ is rather close to $1$, thus showing evidence of potential jump clustering and persistence of low values (compare with the discussion in Section \ref{sec:alpha_stable}). Together with the rather small value of $\theta$, this also indicates a significant likelihood of large jumps.
As illustrated by Figures \ref{fig:calibrationResults}, the model achieves a good fit to market data, across different strikes and maturities. 
We remark that, in terms of number of parameters, the model under consideration is even more parsimonious than the simple specifications calibrated in \cite{CFGaffine}.

\begin{table}[ht]
\centering
\begin{tabular}{|c|c|c|c|}
 \hline 
 $b$ & $0.05353$ & $\alpha$ & $1.31753$ \\ 
 \hline 
 $\sigma$ & $0.00582$ & $y_0$ & $(0.00495,0.00507)^\top$ \\ 
 \hline 
 $\eta$ & $0.04070$ & $\beta$ & $(9.99999E-4, 0.00340)^\top$ \\ 
 \hline 
 $\theta$ & $0.05070$ & $\mu$ & $(1.49999,1.00000)^\top$ \\ 
 \hline 
 \end{tabular}  
\caption{Calibrated values of the parameters. \label{table:calibratedParameters}}
\end{table}

\begin{figure}[ht]
  \centering
   \subfloat{\label{fig:priceCompare}\includegraphics[scale=0.45]{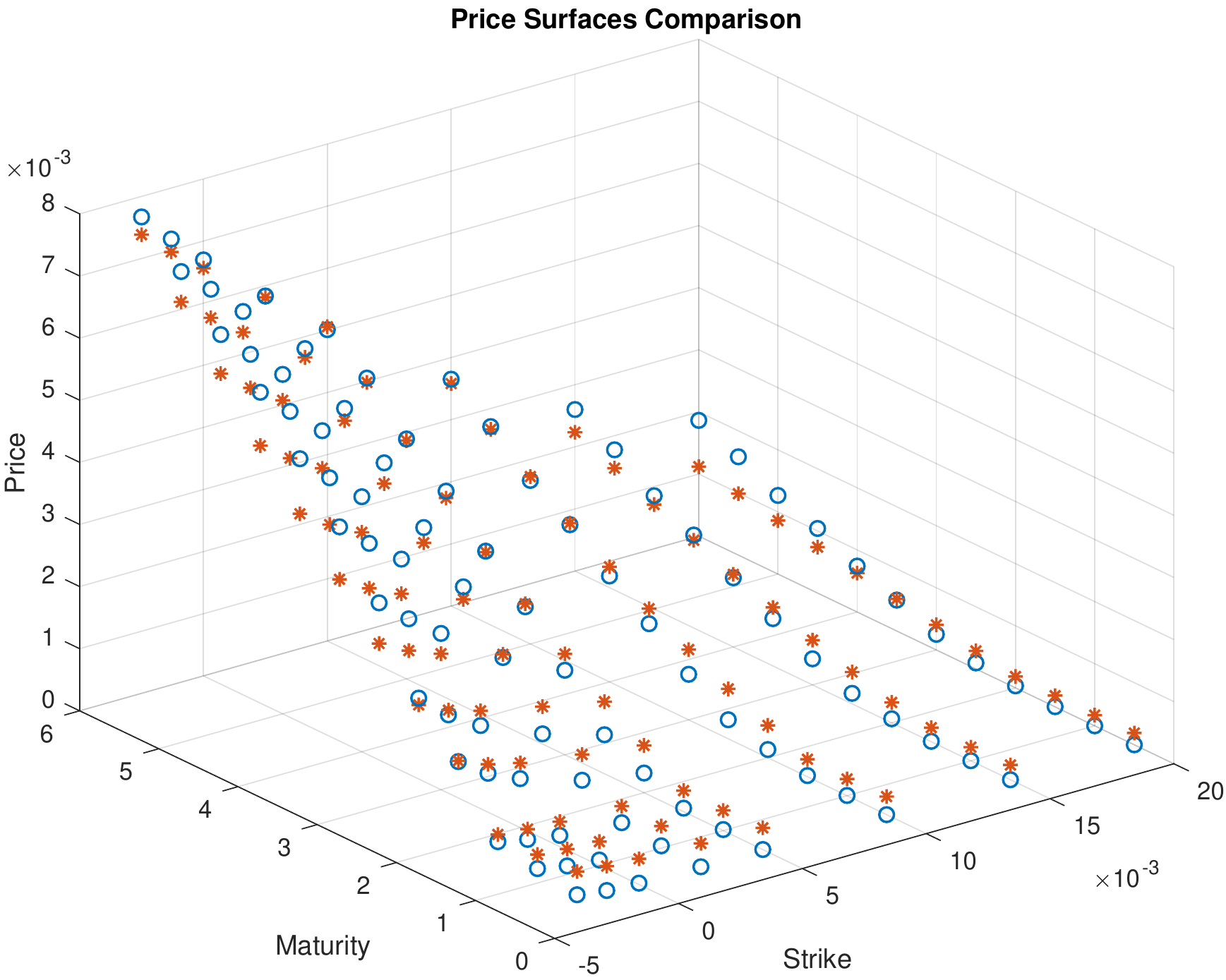}} \
    \subfloat{\label{fig:errorCompare}\includegraphics[scale=0.45]{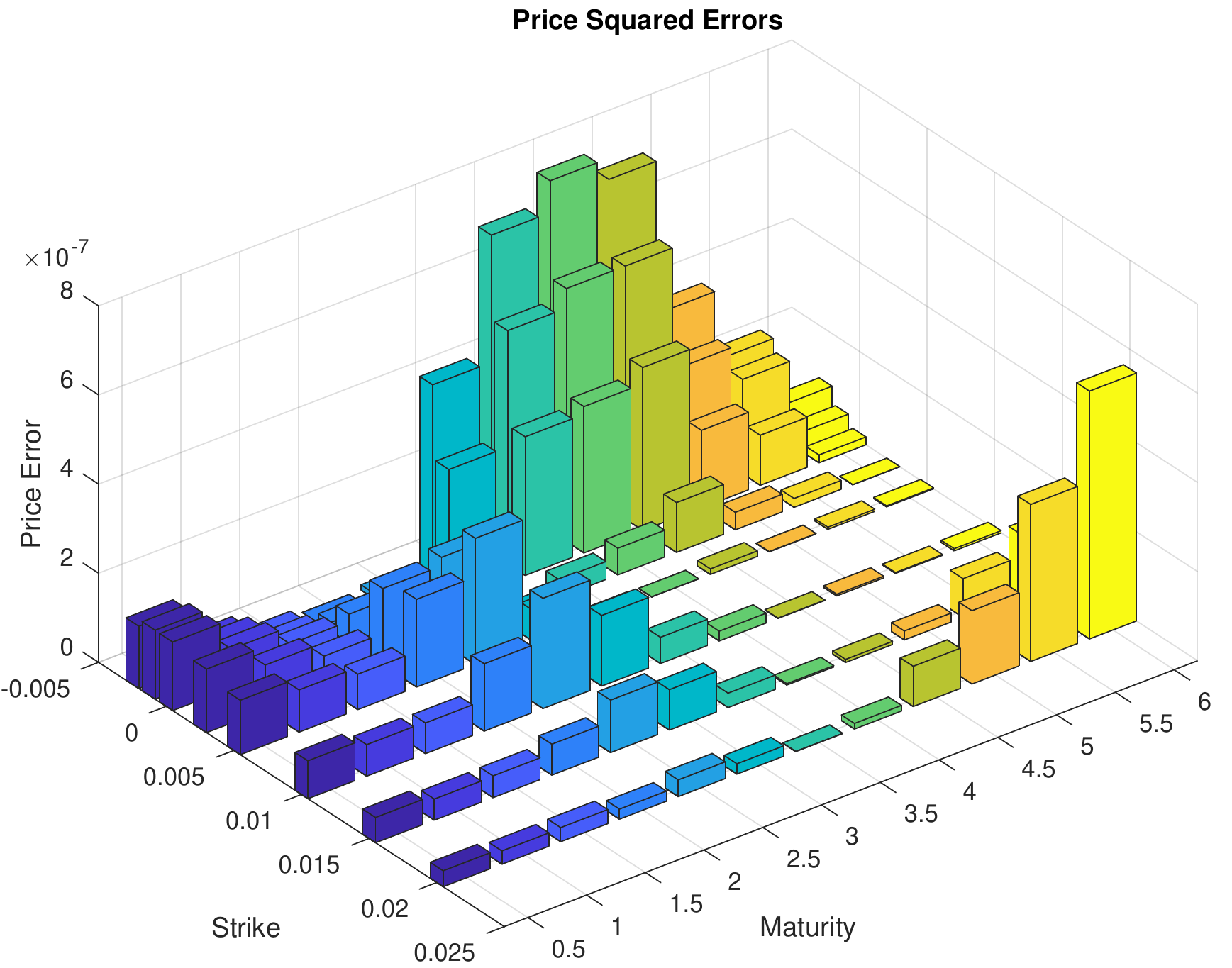}} 
\caption{Model prices against market prices as of 25 June 2018. On the left panel, market prices are represented by blue circles while model prices by red stars.
On the right panel, price squared errors are reported.
\label{fig:calibrationResults}}
\end{figure}

Motivated by the presence of forward rates, we also calibrated a version of the model where the OIS short rate is affected by an auxiliary Ornstein-Uhlenbeck process, in line with Remark \ref{rem:negativeOIS}. However, this alternative specification did not yield a significant improvement of the quality of the fit. This seems to indicate that the deterministic shift $\ell(t)$ introduced in \eqref{eq:flow_OIS} does suffice to capture the probability mass in the negative axis for the short rate. This is also in line with the widespread use of deterministic shift extensions in the financial industry (see, e.g., \cite{mer_futures}).

\section{Conclusions}	\label{sec:conclusions}

In the present paper, we have proposed a modelling framework for multiple yield curves based on CBI processes. The self-exciting behavior of jump-type CBI processes is consistent with most of the empirical features of spreads. At the same time, exploiting the fundamental link with affine processes, our setup allows for an efficient valuation of interest rate derivatives. Models driven by a flow of tempered alpha-stable CBI processes represent a parsimonious way of modelling spreads in a realistic way, with a natural interpretation of the stochastic drivers in terms of risk factors.
In our view, flows of CBI processes can have interesting applications to financial markets where multiple term structures coexist, such as multi-currency or energy markets.

\appendix

\section{Proofs of the Results on CBI Processes}		\label{sec:proofs}

In this appendix, we collect the proofs of the theoretical results stated in Sections \ref{sec:prelim_CBI} and \ref{sec:alpha_stable}.

\begin{proof}[Proof of Theorem \ref{thm:affine}]
The fact that $X$ is a regular affine process follows from \cite[Corollary 2.10]{dfs03}. 
For $p\geq0$ and $q=0$, formula \eqref{eq:aff_transform} simply follows from Definition \ref{def:CBI} and $T^{(p,0)}=+\infty$.
Under Assumption \ref{ass:Lipschitz}, $\phi$ is a locally Lipschitz continuous function on $\cY$. Therefore, for every $(p,q)\in\cY\times\R_+$, standard existence and uniqueness results for solutions to first-order ODEs imply the existence of a maximal lifetime $T^{(p,q)}\in(0,+\infty]$ such that \eqref{eq:ODE_v} admits a unique solution $v(\cdot,p,q):[0,T^{(p,q)})\rightarrow\cY$ and the integral $\int_0^t\psi(v(s,p,q))\ud s$ is finite, for all $t<T^{(p,q)}$.
Hence, part (b) of \cite[Theorem 2.14]{krm12} applied to the bi-dimensional affine process $(X,\int_0^{\cdot}X_s\,\ud s)$ implies that the affine transform formula \eqref{eq:aff_transform} holds for every $(p,q)\in\cY\times\R_+$ and $t<T^{(p,q)}$.
\end{proof}

\begin{proof}[Proof of Theorem \ref{thm:lifetime}]
Let $(p,q)\in\cY\times\R_+$. 
In the trivial case $\phi\equiv0$, the ODE \eqref{eq:ODE_v} is solved by the function $v(t,p,q)=p+qt$ and hence $T^{(p,q)}=+\infty$. 
In the rest of the proof, we shall assume that $\phi(y)\neq0$ for some $y\in\cY$.
Note that $\{y\in\cY:q-\phi(y)\geq0\}\cap\R_-\neq\emptyset$, so that $p_q$ is always well-defined with values in $[\ell\vee\kappa,0]$ and,  by continuity of $\phi$, it satisfies $\phi(p_q)\leq q$.
If $\phi(p_q)=q$, then the constant function $\tilde{v}(\cdot):=p_q$ is a solution to \eqref{eq:ODE_v} with initial value $p=p_q$. Since the ODE \eqref{eq:ODE_v} admits a unique solution for every $(p,q)\in\cY\times\R_+$ by Assumption \ref{ass:Lipschitz}, it holds that $v(t,p_q,q)=p_q$, for all $t\geq0$. Since $p_q\in\cY$, it follows that $T^{(p_q,q)}=+\infty$.
On the other hand, if $\phi(p_q)<q$, this means that $p_q=\ell\vee\kappa$. Let us define $p^+_q:=\sup\{y\in\cY : q-\phi(y)\geq0\}$. By convexity of $\phi$, it holds that $q-\phi(y)>0$ for all $y\in[p_q,p^+_q)$ and, therefore, the function $t\mapsto v(t,p_q,q)$ is increasing.
The ODE \eqref{eq:ODE_v} implies that
\[
\int_{p_q}^{v(t,p_q,q)}\frac{\ud y}{q-\phi(y)} = t,
\qquad\text{ for all }t\geq0.
\]
Letting $t\rightarrow+\infty$ on both sides of the above equality, we get that $v(t,p_q,q)\rightarrow p_q^+$ as $t\rightarrow+\infty$, while $v(t,p_q,q)<p_q^+$ for all $t\geq0$. Since $\psi$ is increasing, it holds that $\psi(p_q)\leq\psi(v(t,p_q,q))\leq\psi(p_q^+)$ and, therefore, the integral  $\int_0^T\psi(v(s,p,q))\ud s$ is finite for all $T\in\R_+$.
This implies that $T^{(p_q,q)}=+\infty$.
By \cite[Theorem 2.14]{krm12} applied to the bi-dimensional affine process $(X,\int_0^{\cdot}X_s\,\ud s)$, this means that $\EE[\exp(-p_qX_t-q\int_0^tX_s\,\ud s)]<+\infty$, for all $t\geq0$. Therefore, it holds that $\EE[\exp(-pX_t-q\int_0^tX_s\,\ud s)]<+\infty$ for all $t\geq0$ and $p\geq p_q$. In turn, by \cite[Proposition 3.3]{krm12}, this implies that $T^{(p,q)}=+\infty$, for all $p\geq p_q$.
Let us now consider the case $p<p_q$ and suppose first that $\kappa\leq\ell$. In this case, due to the convexity of $\phi$, it holds that $q-\phi(y)<0$ for all $y\in[\ell,p_q)$.
Arguing similarly as in \cite[Theorem 4.1]{kr11}, the ODE \eqref{eq:ODE_v} admits a unique solution $v(t,p,q)$ which is strictly decreasing in $t$, with values in $[\ell,p]$. This solution admits a maximal extension to an interval $[0,T^*)$ such that one of the following two situations occurs:
\begin{itemize}
\item[(i)] $T^*=+\infty$;
\item[(ii)] $T^*<+\infty$ and $\lim_{t\rightarrow T^*}v(t,p,q)=\ell$.
\end{itemize} 
In case (i), since $v(\cdot,p,q)$ is strictly decreasing, $\alpha:=\lim_{t\rightarrow+\infty}v(t,p,q)$ is well-defined, with values in $\{-\infty\}\cup[\ell,p)$. If $\alpha>-\infty$, then $\alpha$ must be a stationary point, i.e., $q-\phi(\alpha)=0$. However, this contradicts the fact that $\alpha<p<p_q$. The case $\alpha=-\infty$ can only happen if $\ell=-\infty$ and in this case $\lim_{t\rightarrow T^*}v(t,p,q)=\ell$, exactly as in case (ii).
In case (ii), let $(T_n)_{n\in\N}$ be an increasing sequence with $T_n<T^*$, for all $n\in\N$, such that $T_n\rightarrow T^*$ as $n\rightarrow+\infty$. Similarly as above, it holds that
\be	\label{eq:lifetime_proof}
\int_{p}^{v(T_n,p,q)}\frac{\ud y}{q-\phi(y)} = T_n,
\qquad\text{ for each }n\in\N.
\ee
Letting $n\rightarrow+\infty$ on both sides of \eqref{eq:lifetime_proof}, we obtain that $T^*=\int_p^{\ell}(q-\phi(y))^{-1}\ud y$. If $\kappa\leq\ell$, then $\int_0^t\psi(v(s,p,q))\ud s$ is always finite whenever $v(t,p,q)$ is finite, so that $T^{(p,q)}=T^*$, thus proving \eqref{eq:lifetime} in the case $\kappa\leq\ell$.
If $\kappa>\ell$, then the lifetime is given by $T^{(p,q)}=\inf\{t\in\R_+ : v(t,p,q)=\kappa\}$. Replacing $T_n$ with $T^{(p,q)}$ into \eqref{eq:lifetime_proof} yields \eqref{eq:lifetime}, thus proving the first part of the theorem.

To prove the last statement of the theorem, suppose that $\psi(\ell\vee\kappa)>-\infty$. In this case, if $-\infty<\phi(\ell\vee\kappa)\leq0$, then $\cY=[\ell\vee\kappa,+\infty)$ and $p_q=\ell\vee\kappa$, for every $q\in\R_+$. By the first part of the theorem, it follows that $T^{(p,q)}=+\infty$ for all $(p,q)\in[\ell\vee\kappa,+\infty)\times\R_+$. 
Conversely, if $T^{(p,q)}=+\infty$ for all $(p,q)\in[\ell\vee\kappa,+\infty)\times\R_+$, then in particular $\ell\vee\kappa\in\cY$ and $T^{(\ell\vee\kappa,0)}=+\infty$. This is only possible if $-\infty<\phi(\ell\vee\kappa)\leq0$. Indeed, if $\kappa\leq\ell$ and $\phi(\ell)>0$, then the solution $v(t,\ell,0)$ to the ODE \eqref{eq:ODE_v} with $p=\ell$ would explode immediately (i.e., $T^{(\ell,0)}=0$). Similarly, if $\kappa>\ell$ and $\phi(\kappa)>0$, then the solution $v(t,\kappa,0)$ to \eqref{eq:ODE_v} with $p=\kappa$ would be strictly decreasing in a neighborhood of zero and, therefore, the integral $\int_0^{\cdot}\psi(v(s,\kappa,0))\ud s$ would immediately diverge to $-\infty$ (i.e., $T^{(\kappa,0)}=0$). 
\end{proof}

\begin{proof}[Proof of Lemma \ref{lem:branch_alpha}]
We only consider the case $\theta>0$, referring to \cite{JMS17} for the case $\theta=0$.
Note first that
\[
\int_0^{+\infty}(e^{-zu}-1+zu)\frac{e^{-\theta u}}{u^{1+\alpha}}\ud u
= \int_0^{+\infty}\sum_{n=2}^{+\infty}\frac{(-zu)^n}{n!}u^{-1-\alpha}e^{-\theta u}\ud u.
\]
If $z>-\theta$, we can interchange the order of integration and summation, thus obtaining
\begin{align*}
\int_0^{+\infty}(e^{-zu}-1+zu)\frac{e^{-\theta u}}{u^{1+\alpha}}\ud u
&= \sum_{n=2}^{+\infty}\frac{(-z/\theta)^n}{n!}\theta^{\alpha}\Gamma(n-\alpha)
\\
&= \theta^{\alpha}\Gamma(-\alpha)\left(\frac{\alpha(\alpha-1)}{2!}\left(\frac{z}{\theta}\right)^2+\frac{\alpha(\alpha-1)(\alpha-2)}{3!}\left(\frac{z}{\theta}\right)^3+\dots\right).
\end{align*}
The last line of the above expression is related to the power series
\[
\left(1+\frac{z}{\theta}\right)^{\alpha}
= 1 + \alpha\frac{z}{\theta} + \frac{\alpha(\alpha-1)}{2!}\left(\frac{z}{\theta}\right)^2+\frac{\alpha(\alpha-1)(\alpha-2)}{3!}\left(\frac{z}{\theta}\right)^3+\dots,
\]
which converges if and only if $z>-\theta$. Therefore, it holds that
\[
\int_0^{+\infty}(e^{-zu}-1+zu)\frac{e^{-\theta u}}{u^{1+\alpha}}\ud u
=  \theta^{\alpha}\Gamma(-\alpha)\left(\left(1+\frac{z}{\theta}\right)^{\alpha}-1-\alpha\frac{z}{\theta}\right),
\]
from which \eqref{eq:branch_alpha} follows due to the definition of $C(\alpha,\eta)$ given in \eqref{eq:C}.
By continuity, formula \eqref{eq:branch_alpha} can then be extended to $z=-\theta$. The convexity of $\phi$ follows by noting that
\[
\phi''(z) = \sigma^2 - \eta^{\alpha}\frac{\alpha(\alpha-1)(z+\theta)^{\alpha-2}}{\cos(\alpha\pi/2)} \geq 0,
\qquad\text{ for all }z\geq-\theta.
\]
By computing $\partial\phi(z)/\partial\theta$ and using Bernoulli's inequality, it can be easily verified that $\phi$ is decreasing with respect to $\theta$.
Finally, since $\phi\in\mathcal{C}^1([-\theta,+\infty))$, Assumption \ref{ass:Lipschitz} is satisfied.
\end{proof}

\begin{proof}[Proof of Proposition \ref{prop:properties}]
(i): 
this is a direct consequence of the last part of Theorem \ref{thm:lifetime} together with \eqref{eq:lifetime_moments}.
(ii): as a consequence of \eqref{eq:branch_alpha}, it holds that $\phi(z)\geq bz+\sigma^2z^2/2$, for all $z\geq0$. Furthermore, if $2\beta\geq\sigma^2$, it can be checked that $\psi(z)/\phi(z)\geq z^{-1}(1+O(z^{\alpha-2}))$ for all sufficiently large $z$.
The result follows by the same arguments used in the proof of \cite[Proposition 3.4]{JMS17}.
\end{proof}

\begin{proof}[Proof of Proposition \ref{prop:ergodic}]
The fact that $(P_t(\cdot,x))_{t\geq0}$ converges weakly to a stationary distribution $\rho$ follows from \cite[Corollary 3.21]{Li}, while formula \eqref{eq:ergodic_laplace} for $p\geq0$ corresponds to \cite[Theorem 3.20]{Li}. Consider the case $p\in(p_0,0)$, with $p_0<0$. Since $\phi(z)<0$ for all $z\in(p_0,0)$, the solution $v(t,p,0)$ to the ODE \eqref{eq:ODE_v} with $q=0$ is strictly increasing. Furthermore, \eqref{eq:ODE_v} implies that
\[
-\int^{v(t,p,0)}_p\frac{\ud y}{\phi(y)} = t,
\qquad\text{ for all }t\geq0.
\]
Therefore, letting $t\rightarrow+\infty$ on both sides, it follows that $\lim_{t\rightarrow+\infty}v(t,p,0)=0$. In turn, as a consequence of \eqref{eq:aff_transform} (with $q=0$), this implies that
\[
\lim_{t\rightarrow+\infty}\EE[e^{-pX_t}] 
= \exp\left(-\int_0^{+\infty}\psi(v(s,p,0))\ud s\right)
= \exp\left(-\beta\int_0^p\frac{z}{\phi(z)}\ud z\right),
\]
where the last equality follows by a change of variable together with equation  \eqref{eq:ODE_v}.
Formula \eqref{eq:ergodic_mean} follows by differentiating \eqref{eq:ergodic_laplace} at $p=0$. 
Finally, to prove the exponential ergodicity of $X$, recall that $\phi(z)\geq bz+\sigma^2z^2/2$, for all $z\geq0$ (see the proof of Proposition \ref{prop:properties}). Therefore, it holds that
\[
\int_c^{+\infty}\frac{1}{\phi(z)}\ud z
\leq \int_c^{+\infty}\frac{1}{bz+\frac{\sigma^2z^2}{2}}\ud z < +\infty,
\qquad\text{ for any $c>0$.}
\]
In view of \cite[Theorem 2.5]{LM15} (see also \cite[Theorem 10.5]{Li19}), this suffices to prove the claim.
\end{proof}

\section{Simulation of Tempered Alpha-Stable CBI Processes}	\label{app:simulation}

In this appendix, we present a simulation method for tempered $\alpha$-stable CBI processes, as considered in Section \ref{sec:alpha_stable}. This method has been used to generate the sample paths shown in Figure \ref{fig:model_paths}. 
Note that, in view of Proposition \ref{prop:flow_CBI}, to simulate a multi-curve model driven by a flow of tempered $\alpha$-stable processes as described in Section \ref{sec:model_spec}, it suffices to simulate $m$ independent tempered $\alpha$-stable CBI processes. Therefore, we shall only study the simulation of a one-dimensional tempered $\alpha$-stable process $X=(X_t)_{t\in[0,T]}$, for some fixed time horizon $T>0$ and for $\alpha\in(1,2)$.

We recall from Proposition \ref{prop:CBI} and Remark \ref{rem:BM}  the representation
\be	\label{eq:SDE_alphastable}
X_t = x + \int_0^t(\beta-bX_s)\ud s + \sigma\int_0^t\sqrt{X_s}\ud B_s + \int_0^t\int_0^{+\infty}\!\int_0^{X_{s-}}z\tildeM(\ud s,\ud z,\ud u),
\ee
for all $t\in[0,T]$, where $B=(B_t)_{t\in[0,T]}$ is a one-dimensional Brownian motion and $\tildeM(\ud s,\ud z,\ud u)$ is the compensated Poisson random measure with intensity $\ud s\,\pi(\ud z)\ud u$, where, in view of Definition \ref{def:alpha_stable}, the measure $\pi(\ud z)$ is specified as in \eqref{eq:alpha_stable} with the constant $C$ being chosen as in \eqref{eq:C}. 

To simulate a sample path of $X$, we rely on the regular Euler method for stochastic differential equations with jumps described in \cite[Chapter 6]{PBL10}. 
We consider an equidistant partition of the interval $[0,T]$ with $N$ steps (in our application, we chose $N=1000$). Letting $\Delta:=T/N$ and $t_n:=n\Delta$, for all $n=0,1,\ldots,N$, we denote by $\hat{X}=(\hat{X}_{t_n})_{n=0,1,\ldots,N}$ the simulated path of $X$.

As a first step, we approximate the measure $\pi(\ud z)$ by considering the truncated measure $\pi_{\epsilon}(\ud z):=\pi(\ud z)\ind_{\{z\geq\epsilon\}}$, for a sufficiently small $\epsilon>0$ (in our application, we chose $\epsilon=0.001$)\footnote{This truncation of the jump measure $\pi$ serves the achieve integrability, at the expense of eliminating very small jumps. Along the lines of \cite{AR01}, the small jump component can be approximated by introducing a suitably rescaled Brownian motion $B'$, independent of the Brownian motion $B$ appearing in \eqref{eq:SDE_alphastable}.}. 
The total mass of the measure $\pi_{\epsilon}$ is given by
\[
C_{\epsilon} := \pi_{\epsilon}(\R_+) 
= C(\alpha,\eta)\int_{\epsilon}^{+\infty}\frac{e^{-\theta z}}{z^{1+\alpha}}\ud z
= C(\alpha,\eta)\theta^{\alpha}\Gamma(-\alpha,\epsilon\theta)
= \frac{\eta^{\alpha}\theta^{\alpha}}{\cos(\alpha\pi/2)}\frac{\Gamma(-\alpha,\epsilon\theta)}{\Gamma(-\alpha)},
\]
where $\Gamma(-\alpha,\epsilon\theta):=\int_{\epsilon\theta}^{+\infty}u^{-\alpha-1}e^{-u}\ud u$ denotes the incomplete Gamma function (see \cite[Problem 1.10]{Lebedev}). 
For each $n=1,\ldots,N$, we approximate the number of jumps generated by the random measure $M$ in the interval $[t_{n-1},t_n]$  by a random variable $J_n$ following a Poisson distribution with intensity
\be	\label{eq:sim_intensity}
\int_{t_{n-1}}^{t_n}\int_0^{+\infty}\!\int_0^{\hat{X}_{t_{n-1}}}\ud s\,\pi_{\epsilon}(\ud z)\ud u
= -\hat{X}_{t_{n-1}}C_{\epsilon}\,\Delta.
\ee

The random variables representing the sizes of the jumps generated by the random measure $M$ are drawn from a distribution with density $f_{\epsilon}$, where
\be	\label{eq:sim_density}
f_{\epsilon}(z) = \frac{1}{C_{\epsilon}}\frac{\pi_{\epsilon}(\ud z)}{\ud z}
= \frac{1}{\theta^{\alpha}\Gamma(-\alpha,\epsilon\theta)}\frac{e^{-\theta z}}{z^{1+\alpha}}\ind_{\{z\geq\epsilon\}}.
\ee
Observe that
\be	\label{eq:acc_rej}
f_{\epsilon}(z)  \leq \frac{e^{-\epsilon\theta}}{\alpha(\epsilon\theta)^{\alpha}\,\Gamma(-\alpha,\epsilon\theta)}\frac{\alpha\epsilon^{\alpha}}{z^{1+\alpha}}\ind_{\{z\geq\epsilon\}}
= \frac{e^{-\epsilon\theta}}{\alpha(\epsilon\theta)^{\alpha}\,\Gamma(-\alpha,\epsilon\theta)}f^{\rm Par}_{\epsilon,\alpha}(z),
\ee
where $f^{\rm Par}_{\epsilon,\alpha}$ denotes the density function of a Pareto distribution with scale parameter $\epsilon$ and shape parameter $\alpha$. In view of relation \eqref{eq:acc_rej}, we can simulate random variables with density $f_{\epsilon}$ by means of an acceptance-rejection scheme (see, e.g., \cite[Section 1.4]{g19}) based on a Pareto distribution.

To generate the path $\hat{X}=(\hat{X}_{t_n})_{n=0,1,\ldots,N}$, we set $\hat{X}_0:=x$ and iteratively, for all $n=1,\ldots,N$,
\be	\label{eq:sim_path}
\hat{X}_{t_n} := \hat{X}_{t_{n-1}} + \left(\beta-\Bigl(b+\frac{\eta^{\alpha}\theta^{\alpha-1}}{\cos(\alpha\pi/2)}\frac{\Gamma(1-\alpha,\epsilon\theta)}{\Gamma(-\alpha)}\Bigr)\hat{X}_{t_{n-1}}\right)\Delta + \sigma\sqrt{\hat{X}^+_{t_{n-1}}\,\Delta}\,Z_n+\sum_{k=0}^{J_n}\xi_{n,k},
\ee
where $(Z_n)_{n=1,\ldots,N}$ is a sequence of i.i.d. standard Normal random variables, $(J_n)_{n=1,\ldots,N}$ is a sequence of independent random variables such that each $J_n$ follows a Poisson distribution with intensity given by \eqref{eq:sim_intensity} and $(\xi_{n,k})_{n=1,\ldots,N,k\in\N}$ is a family of i.i.d. random variables with common density $f_{\epsilon}$ as given in \eqref{eq:sim_density}.

\bibliographystyle{alpha}
\bibliography{biblio_multi_curve}

\end{document}